\documentclass[aps,pra,preprint,amsmath,amssymb,longbibliography]{revtex4}
\usepackage{graphicx,epsfig,epsf,color,hhline}% Include figure files
\usepackage{dcolumn}% Align table columns on decimal point
\usepackage{bm,amsthm}% bold math
\usepackage{tensor,pbox}
\usepackage[vcentermath]{youngtab} 
\usepackage{braket}
\newcommand{\kett}[1]{\ket{#1}\rangle}
\newcommand{\brakett}[1]{\langle\braket{#1}\rangle}

\newcommand{\bea}{\begin{eqnarray}}
\newcommand{\ea}{\end{eqnarray}}
\newcommand{\eea}{\end{eqnarray}}

\newcommand{\Nb}{{N_\text{bit}}}
\newcommand{\Ng}{{N_\text{g}}}
\newcommand{\nb}{{n_\text{bit}}}
\newcommand{\Nl}{{N_\text{ball}}}
\newcommand{\Nr}{{N_\text{r}}}
\newcommand{\po}[1]{^{(#1)}}

\newtheorem{theorem}{Theorem}[section]
\newtheorem{definition}[theorem]{Definition}

\newtheorem{lemma}[theorem]{Lemma}

\begin{document}
\newcommand{\ri}{ i}
\newcommand{\re}{ e}
\newcommand{\bx}{{\bm x}}
\newcommand{\bd}{{\bm d}}
\newcommand{\be}{{\bm e}}
\newcommand{\br}{{\bm r}}
\newcommand{\bk}{{\bm k}}
\newcommand{\bA}{{\bm A}}
\newcommand{\bD}{{\bm D}}
\newcommand{\bE}{{\bm E}}
\newcommand{\bB}{{\bm B}}
\newcommand{\bI}{{\bm I}}
\newcommand{\bH}{{\bm H}}
\newcommand{\bL}{{\bm L}}
\newcommand{\bR}{{\bm R}}
\newcommand{\bZero}{{\bm 0}}
\newcommand{\bM}{{\bm M}}
\newcommand{\bX}{{\bm X}}
\newcommand{\bn}{{\bm n}}
\newcommand{\bs}{{\bm s}}
\newcommand{\bv}{{\bm v}}
\newcommand{\tbs}{\tilde{\bm s}}
\newcommand{\rSi}{{\rm Si}}
\newcommand{\beps}{\mbox{\boldmath{$\epsilon$}}}
\newcommand{\bGamma}{\mbox{\boldmath{$\Gamma$}}}
\newcommand{\bxi}{\mbox{\boldmath{$\xi$}}}
\newcommand{\rg}{{\rm g}}
\newcommand{\tr}{{\rm tr}}
\newcommand{\xmax}{x_{\rm max}}
\newcommand{\xb}{\overline{x}}
\newcommand{\pb}{\overline{p}}
\newcommand{\ra}{{\rm a}}
\newcommand{\rx}{{\rm x}}
\newcommand{\rs}{{\rm s}}
\newcommand{\rP}{{\rm P}}
\newcommand{\up}{\uparrow}
\newcommand{\down}{\downarrow}
\newcommand{\hc}{H_{\rm cond}}
\newcommand{\kb}{k_{\rm B}}
\newcommand{\cI}{{\cal I}}
\newcommand{\tit}{\tilde{t}}
\newcommand{\cE}{{\cal E}}
\newcommand{\cC}{{\cal C}}
\newcommand{\Ubs}{U_{\rm BS}}
\newcommand{\sech}{{\rm sech}}
\newcommand{\xs}{{x_1,\ldots,x_N}}
\newcommand{\qq}{{\bf ???}}
\newcommand*{\etal}{\textit{et al.}}
\newcommand{\empt}{\mbox{ }}
\newcommand{\tinI}{{\text{\tiny 1}}}
\newcommand{\tinDrei}{{\text{\tiny 3}}}
\newcommand{\dbla}{{\text{\tiny 2,3}}}
\newcommand{\dbll}{{\text{\tiny 1,2}}}
\newcommand{\dble}{{\tiny\mbox{1,2}}}
\newcommand{\trpl}{{\tiny\mbox{1,2,3}}}
\newcommand{\lra}{\leftrightarrow}
\def\vec#1{\bm{#1}}
\def\ket#1{|#1\rangle}
\def\bra#1{\langle#1|}
\def\keps{\bm{k}\boldsymbol{\varepsilon}}
\def\dm{\boldsymbol{\wp}}

\title{Stochastic emulation of quantum algorithms}
\author{Daniel Braun and Ronny M\"uller}
\affiliation{$^1$Eberhard-Karls-Universit\"at T\"ubingen, Institut f\"ur Theoretische Physik, 72076 T\"ubingen, Germany}

\centerline{\today}
\begin{abstract}
  Quantum algorithms profit from the interference of quantum states in
  an exponentially large Hilbert space and the fact that unitary
  transformations on that Hilbert space can be broken down to
  universal gates that act only on one or two qubits at the same time.
  The former aspect renders the direct classical simulation of quantum
  algorithms difficult.  Here we introduce higher-order partial
  derivatives of a probability distribution of particle positions as a
  new object that shares these basic properties of quantum mechanical
  states needed for a quantum algorithm.  Discretization of the
  positions allows one to represent the quantum mechanical state of
  $\nb$ qubits by $2(\nb+1)$ classical stochastic bits. Based on this,
  we demonstrate many-particle interference and
  representation of pure entangled quantum states via derivatives of
  probability distributions and find the universal set of stochastic maps
  that correspond to the quantum gates in a universal gate set.  We 
  prove that the propagation via the stochastic map built from those
  universal stochastic maps reproduces up to a prefactor exactly the
  evolution of the quantum mechanical state with the corresponding
  quantum algorithm, leading to an automated translation of a quantum
  algorithm to a stochastic classical algorithm.   
  We implement several well-known 
  quantum algorithms, analyse the scaling
  of the needed number of realizations with the number of qubits, and highlight the role of destructive interference for the cost of the emulation. 
  Foundational questions raised by the new
  representation of a quantum state are discussed.   
\end{abstract}
%\pacs{03.67.-a, 03.67.Lx, 03.67.Mn }
\maketitle
%\begin{twocolumn}
%\tableofcontents
%\listoffigures

\section{Introduction}
Schr\"odinger is often quoted with the statement ``I would not call [entanglement] {\em one} but rather {\em the} characteristic trait of quantum mechanics, the one that enforces its entire departure from classical lines of thought.'' \cite{schrodinger_discussion_1935}.  Nevertheless, there are many classical examples of 
``entanglement''. Consider for example an elastic membrane of rectangular
shape with lengths $a$, $b$ in directions $x,y$. It has eigenmodes of
vibration that can be found by solving the wave equation with
appropriate boundary conditions for the amplitude field $\psi(\bm x,t)$ with
$\bm x=(x,y)$, $x\in [0,a]$, $y\in [0,b]$.  A general solution  of the
wave equation can be written as superposition, $\psi(\bm
x,t)=\sum_{n,m}c_{nm}(t)\sin(k_{x,n} x)\sin(k_{y,m} y)$, where
$c_{nm}(t)\in\mathbb R$ are arbitrary real coefficients, and the
$k$-vector components are discrete, $k_{x,n}=\pi n/a$, $k_{y,m}=\pi
m/b$.  Clearly, if there is only one component $c_{nm}(t)$, the state is
a product-state, but most states are entangled in the same sense as in
quantum mechanics, i.e.~they cannot be written as a product of two
wavefunctions.  The wavefunction lives in a Hilbert space $\mathbb
R^\infty \times \mathbb R^\infty$, where the two directions $x$, $y$
play the role of two different subsystems.  The membrane is capable of
full interference in this Hilbert space, owing to the positive or
negative values that the wavefunction can take.  However, we have only
two subsystems, which might be extended to three by going to a
vibrating cube, but clearly this is the maximum number of different
subsystems that we can manipulate independently. The fact that in 
principle we can code any state of a quantum computer with an
arbitrary number of qubits in the Hilbert-space of the membrane or the
cube is of no use, as the known quantum algorithms rely on the fact
that they can be decomposed into quantum gates that act at the same
time on one or two qubits only.  \\

An arbitrary number $N$ of subsystems that can be manipulated
independently is available when we consider instead of a wavefunction
a joint probability distribution $P(x_1,\ldots,x_N)$ of $N$ random variables
$x_i$. Product states can then be defined as
$P(x_1,\ldots,x_N)=\prod_{i=1}^N P_i(x_i)$ --- which of course just means
that there are no correlations between the random variables.  Any
probability distribution that cannot be written in this form is
formally equivalent to a pure entangled state in quantum mechanics,
meaning just this: it cannot be written as a product of local
probability distributions, just as in quantum mechanics the wave
function of an entangled state cannot be written as a product of local
wavefunctions.  There is also a (restricted) superposition 
principle: If $p_a(x_1,\ldots,x_N)$ and $p_b(x_1,\ldots,x_N)$ are two
valid probability distribution, so is their convex combination,
$p(x_1,\ldots,x_N)=q p_a(x_1,\ldots,x_N)+(1-q)p_b(x_1,\ldots,x_N)$
with $0\le q\le 1$.  A local probability distribution might be
expanded in an orthonormal set of basis functions $b_i$ (e.g.~for
$x\in[a,b]$, define $b_i(x)=0$ unless $x\in [i \epsilon,
(i+1)\epsilon]$ where $b_i(x)=1$). A scalar product between two probability states denoted as $\kett{p_1}$ and $\kett{p_2}$ can be defined as
$\langle\langle p_1|p_2\rangle\rangle\equiv \int_{a}^b p_1(x)p_2(x)
dx$, and a 1-norm as $|| \,|p_1\rangle\rangle||\equiv \int_{a}^b p_1(x)
dx$, where $\epsilon$ is some chosen width of a basis function, leading to a finite-dimensional Banach space of dimension 
$(b-a)/\epsilon$).  A probability state such as 
$(|00\rangle\rangle+|11\rangle\rangle)/2$ with probability
distribution $p(x_1,x_2)=(b_0(x_1)b_0(x_2)+ b_1(x_1)b_1(x_2))/2$ then
corresponds to perfect classical correlations.  Clearly, here we are
not limited by the number of subsystems --- but we are, of course,
limited by the type of superpositions that we can create: Only convex
combinations are possible, a state as
e.g.~$(|00\rangle\rangle-|11\rangle\rangle)/2$ is in general forbidden
due  
to the need of positivity of the probability distribution.  This
implies that ``probability states'' are not capable of
interference, which would require compensating some positive
amplitudes by negative ones (or, as in quantum mechanics, complex ones
by other complex ones). Such direct manipulation of  probability distributions with
classical stochastic maps is therefore not sufficient for running
quantum algorithms, even though the modular structure of Hilbert space
(or Banach) space is there, as well as a restricted superposition
principle both locally as well as in the full Banach space of all
subsystems.\\

It is known that for pure state quantum computing unbound
entanglement must arise, otherwise the quantum algorithm can be
simulated efficiently classically \cite{Jozsa03}. But creating unbound
entanglement is not sufficient for having a computational advantage
over a classical computer, as is well-known on a much deeper
level through the Gottesmann-Knill theorem
\cite{gottesman_heisenberg_1998}. From the present example, we clearly see 
that interference is a necessary ingredient for a quantum computer
that cannot be efficiently simulated classically (see also
\cite{Braun06,BraunG08}). \\

From the above considerations we can distill the following
necessary requirements for a system that is supposed to function the
way a quantum computer does:
\begin{enumerate}
\item Composition of a system in terms of individual subsystems that can be manipulated independently (in sync with the first criterion of DiVincenzo \cite{divincenzo_physical_2000} for building a quantum computer, i.e.~the need of a scalable number of well-characterized qubits)
\item Interference on the level of a single subsystem
\item A linear vector space with a tensor product structure, allowing
  one to manipulate states by local operations (or by acting on at
  most two subsystems)
\item Interference on the level of the full vector space
\item Physical reality.
\end{enumerate}
The last point becomes obvious when we want to actually build a
computer.  Mathematically, one can easily imagine all sorts of spaces
that satisfy the first four criteria.  But they have
to describe physical reality if we want to do something with them in
real life.  The question is then whether besides the quantum
mechanical wavefunction other objects exist in Nature that allow one
to satisfy all of the above criteria, and hence, ultimately, run
quantum algorithms using those objects.  Surprisingly, the answer is yes.

\section{A new type of ``wavefunction''  and the
  grabit}
The last example already goes a long way in the direction of what we
need.  In particular, with probabilistic bits, and the full
probability distribution propagated by a stochastic map, we
get the same parallelism of propagation in an
exponentially large (Banach)--space that is often quoted as the origin  of
the ``quantum 
supremacy''. Recently, probabilistic bits (``pbits'') and a physical 
realization via spintronics were suggested in \cite{borders_integer_2019}.   The 
only problem with probability distributions is the impossibility of
interference.  So the question 
is:  can we somehow transform the probability distribution in another
way from what we essentially do in quantum mechanics, namely taking
its square root and providing it with a complex phase factor, in order
to get an object that can at least take 
negative values, and still represent a certain physical reality?\\

% One obvious choice would be to subtract an offset from the probability
% distribution.   E.g.~if $p(x)$ is a probability distribution of a
% single system, consider $q(x)=p(x)-1/2$ as a candidate ``wavefunction'' of the
% system. Clearly, it is capable of single-particle interference, as the
% requirement of positivity is lifted.  However, it quickly turns out
% that we quickly run into trouble when
\subsection{Higher-order derivatives of probability distributions}
Consider as a many-particle wavefunction of $N$ particles on the real
line the $N$-th derivative of the probability distribution, i.e.~define
\begin{equation}
  \label{eq:psi}
  \psi(\xs)\equiv \partial_{x_1}\ldots\partial_{x_N}P(\xs),
\end{equation}
where the r.h.s.~will be abbreviated as $\partial^N P$ as well. In particular, a single--particle wavefunction is just the
(1D)-gradient $\psi(x)=\partial_x P(x)$ of a single-particle
probability distribution.  With $x\in[a,b]$, $P(x)$ lives in the
space ${\cal L}^1[a,b]$ of integrable real functions on $[a,b]$, which
implies that  $\psi(x)$ lives in the
space %${\cal L}^2[a,b]$
of twice-integrable real functions on $[a,b]$. 
We can expand $\psi(x)=\sum_n c_n b_n(x)$, where $c_n\in \mathbb R$,
and $b_n(x)$ is a set of basis functions %$\cal
% L^2[a,b]$
of twice-integrable real functions on $[a,b]$. A scalar product of two
wavefunctions $\psi(x)$ and 
$\phi(x)$ can be defined as $\langle \psi|\phi\rangle=\int_a^b
\psi(x)\phi(x)\,dx$, under which we can orthonormalize the set of basis
functions $\{b_n(x)\}$. A norm of the wave-function can be derived as
usual from the scalar product, i.e.~the Hilbert-Schmidt norm
$||\psi||_2=\sqrt{\braket{\psi|\psi}}$, under which we can normalize the
wavefunctions for making a connection to the usual quantum
mechanical wavefunctions. The 1-norm $||\psi||_1=\int_a^b|\psi(x)|\,dx$ will also be used.\\

In the many-body case, if there are no correlations, then
\begin{eqnarray}
  \label{eq:prodpsi}
  \psi(\xs)&=&\partial_{x_1}\ldots\partial_{x_N}P_1(x_1)\ldots
              P_N(x_N)\\
  &=&\prod_{i=1}^N\partial_{x_i}P_i(x_i)=\prod_{i=1}^N\psi_i(x_i)\,,
\end{eqnarray}
where the $\psi_i(x_i)$ are local ``wavefunctions''.  I.e.~in this
case we obtain a product state for the wavefunction rather than just
for the probability distribution. We can arbitrarily superpose such
product states of basis functions, leading to a general state
\begin{equation}
  \label{eq:psigen}
  \psi(\xs)=\sum_{\bm n} c_{\bm n}b_{n_1}(x_1)\cdot\ldots\cdot b_{n_N}(x_N)\,,
\end{equation}
where ${\bm n}\equiv (n_1,n_2,\ldots,n_N)$ and $c_{\bm n}\in\mathbb
R$. One easily checks that 
with this the criteria 1.)-4.) from the Introduction are
satisfied. We will come back to criterion 5.) in more detail later,
but it is clear that the ``wavefunction'' defined in \eqref{eq:psi}
has a clear physical meaning.  First-order spatial derivatives of
concentrations 
enter for example in Fick's law of diffusion, and osmotic pressure is
driven by concentration differences. Reaction-diffusion systems were
analyzed by Turing as a possible mechanism of pattern formation in
biology \cite{a._m._turing_chemical_1952}, and there is even a branch of physics
trying to establish universal computing based on such systems
\cite{gorecki_chemical_2015}. For large particle numbers, concentrations can be seen as proxies of corresponding probabilities.    Higher-order derivatives known to play a role in physics are  
third-order time-derivatives that appear in electrodynamics when
trying to build it entirely on the motion of charged particles
\cite{LandauLifschitz87}.  Also, in quantum mechanics we can get derivatives of
arbitrarily high order from the time-evolution operator $\exp(-i\hat
H t/\hbar)$ if the kinetic energy in the hamiltonian $\hat H$ is
expressed as a second derivative in position-representation, 
but it appears that derivatives of probability distributions of a given large order have not been considered as physical objects in their own right yet.  
Here we show that by
doing so one can emulate any pure-state quantum algorithm as stochastic classical algorithm. Furthermore, this might open new perspectives on the possible physical nature of the quantum mechanical wave function, which almost 100 years after its invention is still subject to debate (see \cite{pusey_reality_2012} and references therein). 

\subsection{Discretization}
For numerical implementations it is necessary to discretize coordinates $x$, $x=i\Delta$, $i\in \mathbb N$, and approximate
derivatives by quotients of finite differences. Hence, for a single particle, we have
$\psi(x)=\partial_x P(x)\simeq (P((i+1)\Delta)-P(i\Delta))/\Delta \equiv
P({i+1})-P(i)$, where in the last step and from now one we set $\Delta=1$.  I.e.~one can code one value of $\psi$
using two values of $P$. For realizing a single continuous bit,
one needs four values of $P$ \footnote{In principle one could do with
  only three values, with one common support at $x_i$, but the two
  values of $\psi$ would then depend on a common value of $p$, which
  is inconvenient.}. Counting the values $i$ from
$i=0$, we can define a two-state system with continuous real
coefficients $\psi_0$ and $\psi_1$ as
\begin{equation}
  \label{eq:psi01}
   {\psi_0 \choose \psi_1}=-{P(x_1)-P(x_0) \choose P(x_3)-P(x_2)}\equiv
   {P_0-P_1 \choose P_2-P_3} \,.
 \end{equation}
 For convenience we have introduced an (irrelevant) global minus-sign.
Apart from the fact that the $\psi_i$ are real, and possibly a different
normalization, the state defined by 
\eqref{eq:psi01} emulates a qubit in a pure state
with amplitudes $\psi_i$ for computational basis states
$\ket{i}$. Owing to the origin as a gradient of the probability
distribution, a system with a two-dimensional state space as defined
in \eqref{eq:psi01} will be called a ``gradient-bit'', or ``grabit''
for short.  Eq.\eqref{eq:psi01} can be written compactly as
\begin{equation}
  \label{eq:psii}
  \psi_i=\sum_{\sigma=0,1}(-1)^\sigma P_{2i+\sigma}, \,\,i\in\{0,1\}\,.
\end{equation}
We will call $i$ the ``byte logical value'' (``blv'' or ``logical value'' for short), as it labels the grabit states   
$\ket{0}$ or $\ket{1}$, $\sigma$ the ``gradient value'', and $I=2i+\sigma$
the ``byte4 value'' (``b4v'' for short). They will be denoted with small roman, small
greek, and capital roman letters, respectively.  We can think of $i$
and $\sigma$ really as two independent coordinates that label four
different bins, see Fig.~\ref{fig:grabitsquare}. A probabilistic
realization of the grabit is realized via a particle that is found
with probability $P_I$ in one of the four bins. If just one of the 4 bins
is realized with probability 1, we have states 
$\ket{0}$, $-\ket{0}$, $\ket{1}$, and $-\ket{1}$, respectively, for
$I=0,1,2,3$, or equivalently $(i,\sigma)=(0,0)$, $(0,1)$,  $(1,0)$,
$(1,1)$. So two classical stochastic bits code one (real) grabit.
\begin{figure}[t]
  \includegraphics[width=0.25\textwidth]{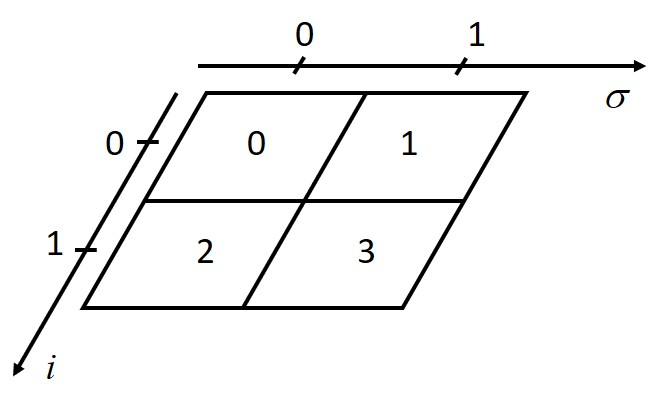}
  \includegraphics[width=0.25\textwidth]{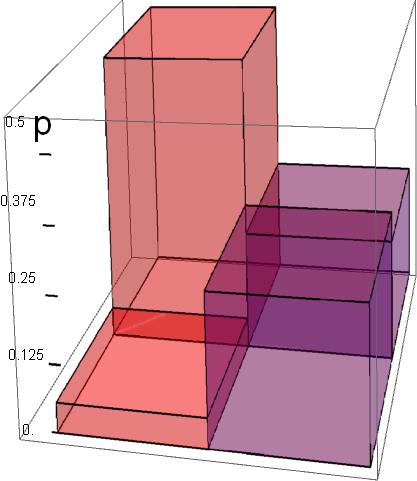}
  \includegraphics[width=0.25\textwidth]{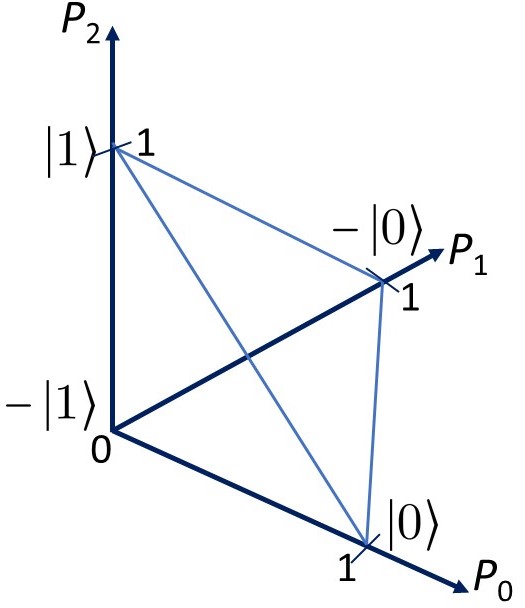}
  \caption{{\em Left:} A single grabit realized as a discretized gradient of a
          probability distribution leads to four bins in which a ball
          is found with
          probabilities $P_I$, $I=0,1,2,3$. The four bins are arranged
          on a square such that $I=2i+\sigma$  where the index $i=0,1$
          labels the logical value and $\sigma=0,1$ the gradient
          value. The (unnormalized) grabit state is given by the
          amplitudes $\psi_0=P_0-P_1$ and
          $\psi_1=P_2-P_3$. {\em Middle:} Example of a grabit state
          $\ket{\psi}=0.3\ket{0}-0.2\ket{1}$, realized via $\bm P=(0.5,0.2,0.05,0.25)$.   
{\em Right:} State space of a single grabit. Allowed pure states are on the surface of the tetrahedron, whose vertices give the signed computational basis states. Convex combinations of the b4v probability states lead to grabit states capable of interference.
          \label{fig:grabitsquare} }
\end{figure}
% We will continue to denote probability distributions over byte4 values
% with capital letters $P_I$; probability distributions over byte
% logical values $i$ will be denoted with small letters, $p_i$.  

\subsection{Single grabit interference}
Let us verify at the hand of the simplest example that interference on
the single-grabit level works correctly, by checking that a state
$\ket{0}$ and a state $-\ket{0}$ can indeed be superposed to give
0$\ket{0}$:
\begin{equation}
  \label{eq:sup}
\ket{0}+(-\ket{0})=0\ket{0} \,.
\end{equation}
According to
\eqref{eq:psii}, the two states involved are represented in terms of
probabilities as $P_0=1$ and $P_1$=1, respectively, or, in terms of
the full 4D probability vectors, as $\bm P_A=(1,0,0,0)^t$, $\bm
P_B=(0,1,0,0)^t$. The linear superposition \eqref{eq:sup} translates
to an equal weight convex combination of the two
probability distributions, $\bm P=(\bm P_A+\bm P_B)/2=(1,1,0,0)/2$.
Inserting this into \eqref{eq:psii} leads immediately to
$\psi_0=\psi_1=0$. So interference on the single-grabit level works
indeed. \\

All other examples of single-grabit interference work the same way:
convex combinations of the probability values $P_I$ ($I=0,1,2,3$)
translate into real superpositions of the $\psi_i$ with coefficients
of either sign and vice
versa. The values of the coefficients in the superposition are only
limited by the fact that $0\le P_I\le 1$ and $\sum_{I=0}^3P_I=1$. Hence,
the state space of the single grabit has the geometry of a
tetrahedron, see Fig.~\ref{fig:grabitsquare}. %\qdq Here , see Fig.~\ref{fig.tetra} was referenced but that figure doesnt exists \qdq
% \begin{figure}[t]
%   \includegraphics[width=0.3\textwidth]{tetrahedron.jpg}
%   \caption{State space of a single grabit. Allowed pure states are on the surface of the tetrahedron, whose vertices give the signed computational basis states. Convex combinations of the b4v probability states lead to grabit states capable of interference.  \label{fig:tetra} }
% \end{figure}
Its vertices correspond to the
states $\ket{0},-\ket{0},\ket{1},-\ket{1}$.  Points on the edges
interpolating between vertices with
different logical values give
linear superpositions of $\ket{0}$ and $\ket{1}$ with signs of the
coefficients given by the signs of the vertices, whereas points on
edges 
interpolating between vertices with
the same logical values still represent states $\ket{0}$ or $\ket{1}$,
respectively, with amplitudes varying from -1 to +1.  Only points on
the surface of the tetrahedron are allowed for pure states.  Since only real linear combinations are 
allowed, the surface of the tetrahedron can be
seen as the analogue of a great circle of the Bloch-sphere in quantum
mechanics in the $xz$-plane.\\

Since in quantum mechanics states are described by points in a
projective Hilbert space (i.e.~state vectors are only defined up to a
complex prefactor), the relevant quantum information in a grabit is
only the ratio of the amplitudes of states $\ket{0}$ and
$\ket{1}$. This motivates a parametrization of the probability vector
$\bm P$ of a single grabit in the form
\begin{equation}
  \label{eq:par}
  \bm P=\frac{1}{2}\begin{pmatrix} 0\\0\\1\\1\end{pmatrix}
  +\frac{r}{2}\begin{pmatrix} 1\\1\\-1\\-1\end{pmatrix}
  +\frac{b}{2}\begin{pmatrix} q\\-q\\1\\-1\end{pmatrix}\,,
\end{equation}
where $\psi_1=P_2-P_3\equiv b\ne 0$ is assumed. The grabit state is
then coded in $q=\psi_0/\psi_1$, whereas $r$ can be seen as a
gauge-degree of freedom that does not change the state of the
grabit.  In the case of $\psi_1=b=0$, we can take $\psi_0=a$ as
independent variable instead of $q$, giving 
\begin{equation}
  \label{eq:par}
  \bm P=\frac{1}{2}\begin{pmatrix} 0\\0\\1\\1\end{pmatrix}
  +\frac{r}{2}\begin{pmatrix} 1\\1\\-1\\-1\end{pmatrix}
  +\frac{a}{2}\begin{pmatrix} 1\\-1\\0\\0\end{pmatrix}\,.
\end{equation}
While $r,b,q$ (or $r$, $a$ for $b=0$) can be seen as new independent
variables that replace $P_0,P_1,P_2$, their ranges are not
independent. Indeed, from $0\le P_I\le 1$, for $I\in\{0,1,2,3\}$, it
follows for $q>0$ that
\begin{equation}
  \label{eq:rgb}
  |b|\le\left\{
      \begin{matrix}
        r/q & 0\le r \le \frac{q}{1+q}\\
        1-r & \frac{q}{1+q}\le r\le 1\,.
      \end{matrix}
    \right.
\end{equation}
This implies that $\text{argmax}_r |b|=q/(1+q)$, and $\text{max}_r
|b|=1/(1+q)$.   The value of $a$ is maximized for the same value of
$r$ and gives $\text{max}_r
|a|=q/(1+q)$. The probability distribution with maximum contrast,
meaning the largest maximum amplitudes $|\psi_0|$ and $|\psi_1|$  for
given value of $q>0$, obtained from $\psi_0>0$ and $\psi_1>0$, is given by
\begin{equation}
  \label{eq:pmax1}
  \bm P_\text{max}(q)=\begin{pmatrix} \frac{q}{1+q}\\0\\\frac{1}{1+q}\\0\end{pmatrix}\,.
\end{equation}
The correct ratio between $\psi_0$ and $\psi_1$ is obvious from this
expression, as is the fact that there is no ``socket'' of
center-of-mass probabilities, in the sense that $P_0-P_1=P_0+P_1$ and
$P_2-P_3=P_2+P_3$, i.e.~all the probability in bins 0 or 1 goes into
the difference that determines $\psi_0$, and all probability in bins 2
or 3 goes into the difference that determines $\psi_1$.  This corresponds 
to a gauge choice, namely the maximum possible value of $r$. We see 
that in this case the {\em physical probability} $\tilde{p}_0=P_0+P_1$ of finding a logical
value 0 (or physical probability $\tilde{p}_1=P_2+P_3$ for finding a
logical value 1) when marginalizing over the gradient value 
is given directly by $\tilde{p}_0=|\psi_0|$ (respectively
$\tilde{p}_1=|\psi_1|$).  This marks a deviation from standard quantum
mechanics and the Born rule that probabilities are given by the {\em
  squared} absolute value of the amplitude, unless a single
probability equals 1.  Useful quantum
algorithms are mostly constructed such that there are one or relatively
few peaks in the final probability distribution on logical values that
allow one to obtain the result of the calculation by efficient
classical post-processing. 
In that case, it is
irrelevant whether the first or second power of the absolute amplitude
determines the probability of the outcome. 
Interestingly, Born
initially advocated the use of the power one \cite{BornBio}.  We will therefore
refer to the rule  ${p}_i=|\psi_i|$ as the ``Born-1'' rule. \\
Most
importantly, however, we see that in this situation of maximum
contrast we have $\tilde{p}_i=p_i\equiv|\psi_i|$, i.e.~the actual physical
probability to find the particle in a bin with logical value $i$ is
given by the Born-1 rule.  This means that if we manage to propagate
$\psi$ according to a quantum algorithm that ends with a peak in
$|\psi|$ AND we can enforce the gauge choice of maximum contrast, the
particle will be found with correspondingly high probability in the
bin that corresponds to the maximum of $|\psi|^2$ --- just as if the
algorithm was run on a quantum computer.\\

For $q<0$ obtained through $\psi_0<0$, $\psi_1>0$,
eq.\eqref{eq:pmax1} reads correspondingly
\begin{equation}
  \label{eq:pmax2}
  \bm P_\text{max}(q)=\begin{pmatrix} 0\\\frac{q}{1+|q|}\\\frac{1}{1+|q|}\\0\end{pmatrix}\,,
\end{equation}
and similarly for the other combinations of signs of $\psi_0$ and
$\psi_1$.  The rule is that always either $P_0$ or $P_1$ equals
$|q|/(1+|q|)$ (and the other one equals $0$) if the sign of $\psi_0$
is positive or negative, respectively; and similarly for $P_2$ and
$P_3$ which always equal $1/(1+|q|)$ or 0, depending on the sign of
$\psi_1$.
\section{Multi-grabit states}
The discretization procedure introduced for a single grabit also works
for any number of grabits. E.g.~for two grabits we have
\begin{eqnarray}
  \label{eq:psi12}
  \psi(x_1,x_2)&=&\partial_{x_1}\partial_{x_2}P(x_1,x_2)\simeq
  \partial_{x_1}(P(x_1,x_2+\Delta x_2)-P(x_1,x_2+\Delta x_2))/\Delta
                   x_2\\
  &\simeq& \frac{1}{\Delta x_1 \Delta x_2}(P(x_1+\Delta x_1,x_2+\Delta x_2)-P(x_1,x_2+\Delta
  x_2)-P(x_1+\Delta x_1,x_2)+P(x_1,x_2)\,.\nonumber
\end{eqnarray}
Setting again $\Delta x_i=1$ and numbering bins as before, this can be
written as
 \begin{equation}
  \label{eq:psi12.2}
  \psi_{i_1,i_2}=\sum_{\sigma_1,\sigma_2\in\{0,1\}}(-1)^{\sigma_1+\sigma_2}P_{2i_1+\sigma_1,2i_2+\sigma_2}\,.
\end{equation}
As in quantum mechanics, the sign of the state can now arise from
either grabit, e.g.~a state $\ket{00}$ can be represented by
$P_{11}=1$, or equivalently by $P_{00}=1$. The latter case is really
a  
product state $(-\ket{0})\otimes(-\ket{0})$.  Two-particle
interference works just as nicely as single-particle interference,
e.g.~$\ket{00}+(-\ket{00})/2$ can be represented by the convex
combination of two probability vectors with components $(\bm P_1)_{IJ}=\delta_{{IJ},00}$,
$(\bm P_2)_{IJ}=\delta_{IJ,01}$, which correspond to states
$\ket{00}$ and $\ket{0}(-\ket{0})$.  The convex combination
$\bm P=(\bm P_1+\bm P_2)/2$ has then components
$P_{IJ}=\delta_{I,0}((\delta_{J,0}+\delta_{J,1})/2)$, which translates
to $\psi_{ij}=(\psi_1)_i(\psi_2)_j$ with $\psi_2=(0,0)^t$, hence
$\ket{\psi}=0$, indeed. \\

We can also easily create entangled states.  It
is useful to visualize 
the addition of
probabilities and resulting grabit states. Let us denote a probability
distribution with just one byte4 value of $I$ realized for a grabit
number 1, e.g.~$I=0$, 
as {\scriptsize{\young(1\empt,\empt\empt)}}, where the squares are
labeled according to Fig.\ref{fig:grabitsquare}. So this state is
$+\ket{0}$ of the first grabit. Similarly, if there are 
two grabits with a single combination of bins realized, we denote it
as e.g.~{\scriptsize{\young(\empt \dble,\empt \empt)}}, which
represents state $(-\ket{0})(-\ket{0})=\ket{00}$, or
{\scriptsize{\young(\empt 1,2\empt)}} which represents the  state
$-\ket{01}$. The Bell states $(\ket{00}\pm\ket{11})/2$ are then realized
e.g.~by 
\begin{eqnarray}
  \label{eq:bellp}
\ket{\psi_+}&=&\frac{1}{2}(\ket{00}+\ket{11})=
                \frac{1}{2}\left({\scriptsize{\young(\dble \empt
    ,\empt \empt)}}+{\scriptsize{\young(\empt
    \empt,\dble \empt )}}\right)\\
\ket{\psi_-}&=&\frac{1}{2}(\ket{00}-\ket{11})=
                \frac{1}{2}\left({\scriptsize{\young(\dble \empt
    ,\empt \empt)}}+{\scriptsize{\young(\empt
    \empt,21)}}\right)\,.
\end{eqnarray}
Two-particle interference works like this:
\begin{eqnarray}
  \label{eq:bellsum}
\frac{1}{2}(\ket{\psi_+}+\ket{\psi_-})&=&\frac{1}{4}\left({\scriptsize{\young(\dble
                                          \empt,\empt \empt)}}+{\scriptsize{\young(\empt
    \empt,\dble \empt)}}+{\scriptsize{\young(\dble \empt,\empt \empt)}}+{\scriptsize{\young(\empt
    \empt,21)}}\right)\\
 &=&\frac{1}{2}{\scriptsize{\young(\dble\empt,\empt \empt)}}
 +\frac{1}{4}{\scriptsize{\young(\empt
    \empt,\dble \empt)}}
     +\frac{1}{4}{\scriptsize{\young(\empt \empt,21)}}\\
                                      &=&\frac{1}{2}\ket{00}+\frac{1}{4}(
                                          \ket{1}+(-\ket{1})\otimes \ket{1}\\
  &=&\frac{1}{2}\ket{00}\,,
\end{eqnarray}
thus giving the correct result: The component $\ket{11}$ has
destructively interfered and vanishes due to the superposition of $\ket{1}$
and $-\ket{1}$ in the first grabit.\\

All of this directly generalizes to any number of grabits. A general
$N$-grabit state is given by 
 \begin{equation}
  \label{eq:psiN}
  \psi_{\bm i}=\sum_{\bm \sigma
    \in\{0,1\}^{\otimes N}}(-1)^{\sum_{i=1}^N\sigma_i}\,P_{2\bm i+\bm \sigma}\,,
\end{equation}
where $\bm i=(i_1,\ldots,i_N)$ and $\bm
\sigma=(\sigma_1,\ldots,\sigma_N)$ label the logical values and
gradient values $0,\ldots
2^{N}-1$, respectively. The map
\eqref{eq:psiN}, $P\mapsto \psi$, will be denoted symbolically as
$\psi=\partial^{N}P$, recalling the origin as a (discretized) $N$th order derivative
of a probability distribution. \\

% A generalization of \eqref{eq:bellsum} to
% states 
% $\ket{\psi}_{\pm}^{(3)}=(\ket{000}\pm\ket{111})/2$ is then realized
% e.g.~by 
% \begin{eqnarray}
%   \label{eq:bell3sum}
%   \frac{1}{2}(\ket{\psi_+}^{(3)}+\ket{\psi_-}^{(3)})
%   &=&\frac{1}{2}\left(\frac{1}{2}\left(
%   {{\young(\trpl\empt,\empt\empt)}}+
%   {{\young(\empt\empt,\trpl\empt)}}\right)+
%       \frac{1}{2}\left(
% {{\young(\trpl\empt,\empt\empt)}}+
% {{\young(\empt\empt,\dbll\tinDrei)}}\right)\right)\\
%   &=&\frac{1}{2}
% {{\young(\trpl\empt,\empt\empt)}}
%       +\frac{1}{4}
% {{\young(\empt\empt,\trpl\empt)}}
%       +\frac{1}{4}
% {{\young(\empt\empt,\dbll\tinDrei)}}\\
%   &=&\frac{1}{2}\ket{000}+\frac{1}{2}(\ket{11})\frac{1}{2}\cdot 0 \cdot\ket{1}\\
%   &=&\frac{1}{2}\ket{000}\,.
% \end{eqnarray}
% % \stopped here

To summarize, a $4^N$-component probability vector is used for coding
the $2^N$ components of a real-valued $\psi$ (see
Sec.\ref{Sec.complex} for complex states).  This means that two classical bits %ref here is weird in my compilation
are required to represent one real-valued qubit.  As in standard
quantum mechanics, where the quantum state is
never explicitly stored but miraculously kept track of by Nature, the
probability vector is never stored explicitly either. 
It is
implicitly stored in the correlations  of the particles
realizing different grabits, or, more precisely, in the derivatives of
order $\Nb$ for $\Nb$ grabits of those correlations.  The derivatives are
approximated as finite differences of those correlations.
In order to emulate a quantum algorithm, we have to sample from these
probability distributions and then act on the realizations in order
to emulate quantum gates for implementing the corresponding stochastic
maps, as will be seen in the next
section. Sampling from the probability distributions gives us $\Nl$
realizations of a set of $\Nb$ little billard balls.  % Each billard
% ball is 
% in one of
% the four states 0,1,2,3 and represents one grabit.  
When $\Nl$ is large
enough, we can estimate 
the probability distribution $P$ over the byte4 values from the observed
relative frequencies $R$.  We will continue to use upper case $P$ (or $P_I$ for a
single component) for probability vectors over the byte4 values
$I$. Column vectors of components $P_I$ may also be represented in
vector notation, $\bm P$.  
Relative frequencies (i.e.~histograms) over b4v's will be denoted by upper case
$R_I$, relative frequencies over blv's by lower case $r_i$. Instead of $i$ and $I$ running from $0$ to $2^\Nb-1$ or $4^\Nb-1$, respectively, we also use $\bm i$ and $\bm I$ when we think of them as strings of single grabit vlaues, $\bm i=(i_1,\ldots,i_\Nb)$ and $\bm I=(I_1,\ldots,I_\Nb)$, and correspondingly for $\bm \sigma$.  All 
probabilities and relative frequencies will be normalized to
1, $\sum_{\bm i\in \{0,1\}^\Nb} r_{\bm i}=1$ and   $\sum_{\bm I\in
    \{0,1,2,3\}^\Nb} R_{\bm I}=1$. The symbol $\psi$ will be used for a grabit state (see
    eq.\eqref{eq:psiN}).  A hat is reserved for estimators.  In
    particular, an estimator of a real-valued multi-grabit state $\psi$ based on relative frequencies
    is
    \begin{equation}
      \label{eq:psiest}
      \hat{\psi}_{\bm i}=\sum_{\bm\sigma\in\{0,1\}^\Nb} (-1)^{\sum_{k=1}^\Nb
        \sigma_k}R_{2\bm i+\bm \sigma}\,.
    \end{equation}
    We also use histograms $H_{\bm I}$ normalized to $\sum_{\bm I}H_{\bm I}=\Nl$. 
The physical probability distribution is the one with the gradient degree of freedom ``traced-out'' (or marginalized), 
    \begin{equation}
      \label{eq:ptilde}
      \tilde{p}_{\bm i}=\sum_{\bm\sigma\in\{0,1\}^\Nb} P_{2\bm i+\bm \sigma}\,,
    \end{equation}
    and its estimator based on the relative frequencies reads
    \begin{equation}
      \label{eq:ptildeest}
      \hat{\tilde{p}}_{\bm i}=\sum_{\bm\sigma\in\{0,1\}^\Nb} R_{2\bm i+\bm \sigma}\,,
    \end{equation}
    $\Psi$ will be reserved for the true
quantum mechanical state with complex amplitudes. \\

We note in passing that there are other ways of defining a
many-body state capable of interference in a high-dimensional Banach
space.  E.g.~one can introduce
an extra coordinate $y_i$. With a pair of
coordinates $x_i,y_i$, single-particle probability distribution
$p(x_i,y_i)$ for particle $i$, and a joint probability distribution
$P(x_1,y_1,\ldots,x_N,y_N)$, we can then define a many-body state
 as
$\psi(x_1,\ldots,x_N)=\partial_{y_1}\ldots\partial_{y_N}P(x_1,y_1,\ldots,x_N,y_N)|_{\bm
  y=\bm 0}$.  Indeed, the gradient value $\sigma_i\in\{0,1\}$ introduced
in the discretization scheme above plays exactly the role of such an
additional coordinate for particle $i$.  Such an additional coordinate
leads to interesting questions about physical reality (see
\ref{Sec.Realism}).  

\section{Grabit gates}
As we have seen, superpositions with real coefficients of either sign of
states of $N$ grabits are generated through convex combinations of
probability vectors over the $4^N$ byte4 values. Such convex
combinations are created physically as usual via stochastic
maps. Physically, we will act on the drawn realizations in
such a way that the derivatives (respectively, differences under the
approximations made) of the  probability distributions are propagated
in the same way as the multi-particle quantum mechanical wave function. In
the special case of stochastic maps occurring with probability one, we
have deterministic gates. Let us now determine the stochastic gates
corresponding to the most important single-qubit and two-qubit gates.

\subsection{Single grabit gates}
\subsubsection{The $X$ gate}
The $X$-gate flips the logical value $i$ while keeping the phase,
i.e.~the gradient value $\sigma$. I.e.~it maps
$\pm\ket{0}\leftrightarrow\pm\ket{1}$. This is achieved by flipping with
probability 1 the byte4 values $0\leftrightarrow 2$ and
$1\leftrightarrow 3$, which translates to a physical operation of
flipping the corresponding particle positions.  Hence, the
stochastic map of a 
single grabit's vector of probabilities $\bm P=(P_0,P_1,P_2,P_3)^T$ over
the four byte4 values is given by 
\begin{equation}
  \label{eq:Sx}
  S_x=\begin{pmatrix}
    0&0&1&0\\
    0&0&0&1\\
    1&0&0&0\\
    0&1&0&0
    \end{pmatrix}\,.
  \end{equation}
  Note that, just as for qubits, it is enough to specify the
  stochastic map for computational basis states.  Due to the
  linearity of the stochastic propagation, it then works for any grabit
  state.
  
\subsubsection{The $Z$ gate}
The $Z$ gate flips the phase of the $\ket{1}$-state and leaves
the phase of $\ket{0}$ untouched,
i.e.~$\pm\ket{0}\lra\pm\ket{0}$ and  $\pm\ket{1}\lra\mp\ket{1}$. In
terms of the particle's bins, this corresponds to flipping the byte4
values as $0\lra 0$, $1\lra 1$, $2\lra 3$, $3\lra 2$, all with
probability 1.  Hence, the stochastic map reads
\begin{equation}
  \label{eq:Sz}
  S_z=\begin{pmatrix}
    1&0&0&0\\
    0&1&0&0\\
    0&0&0&1\\
    0&0&1&0
    \end{pmatrix}\,,
  \end{equation}
  which happens to be the same as the representation of  the ordinary CNOT for two qubits.
\subsubsection{The Hadamard gate}
The Hadamard gate is maybe the most important gate of all.
Almost all known quantum algorithms begin by applying the Hadamard
gate to all qubits.  It is the decisive gate for creating the massive
interference that characterizes a quantum computer.  A single Hadamard
gate creates {\em one i-bit} of interference \cite{Braun06,Arnaud07,BraunG08}
by realizing the transformations $\pm\ket{0}\lra
\pm(\ket{0}+\ket{1})/2$ and $\pm\ket{1}\lra
\pm(\ket{0}-\ket{1})/2$. Here we have already used a normalization that is
adapted to propagating probabilities, but as mentioned earlier, only
ratios between amplitudes matter, such that the normalization by
dividing by 2 instead of the usual $\sqrt{2}$ is irrelevant. The
mapping of states implies that we have to flip byte4 
values for the particle according to $0\to 2$ or $0$, $1\to 3$ or
$1$, $2\to 3$ or $0$, $3\to 1$ or $3\to 2$ with all flips
performed with probability $1/2$. The stochastic map that describes
this propagation reads    
\begin{equation}
  \label{eq:SH}
  S_H=\frac{1}{2}\begin{pmatrix}
    1&0&1&0\\
    0&1&0&1\\
    1&0&0&1\\
    0&1&1&0
    \end{pmatrix}\,.
\end{equation}
We can see that $S_H=(S_X+S_Z)/2$, just as for unitary single-qubit
gates $H=(X+Z)/\sqrt{2}$. \\
A universal quantum computer can be constructed once we can implement
the Hadamard gate, the $\pi/8$-gate (the unitary
$\text{diag}(1,e^{i\pi/4})$ in the computational basis), and a
controlled-NOT (``CNOT'') gate.  A single-grabit
gate representing a unitary with 
finite imaginary parts can be realized with a stochastic map acting
on two grabits at 
the time, and also the CNOT is a two-qubit gate to which we turn now.

\subsection{Two-grabit gates}
\subsubsection{The CNOT}\label{Sec.CNOT}
The CNOT flips a target qubit if a control qubit is in state
$\ket{1}$. I.e.~if we have two qubits, with the first one %\qdq being
the
control-qubit and the second one %\qdq being  braucht man hier nicht...
the target-qubit, then the CNOT performs
a permutation of computational basis states according to
$\pm\ket{00}\lra\pm\ket{00}$, $\pm\ket{01}\lra\pm\ket{01}$,
$\pm\ket{10}\lra\pm\ket{11}$. This is easily translated to a
stochastic map for byte4 values of grabits by keeping the local signs,
just as in the case of the single-grabit gates. For the CNOT this translates
to a 16$\times$16 permutation matrix $S_\text{CNOT}$ that will not be written down here, but as 
an example,  $(-\ket{1})(+\ket{1})\lra(-\ket{1})(+\ket{0})$
translates to flipping joint b4v's of particle positions $(3,2)\lra
(3,0)$ with probability 1.

\subsubsection{Complex states and gates}\label{Sec.complex}
Up to now we have considered interference only on the basis of
positive and negative amplitudes. However, the state space of quantum
mechanics is a (projective) Hilbert-space over the complex numbers.
It is well-known that quantum mechanics can be formulated
completely equivalently with state vectors (and hence wavefunctions)
over the field of real numbers only (see however \cite{wu_operational_2021} for recent work considering imaginarity as a resource).  To this end, one splits the wave
function into its real and imaginary part, i.e.~it becomes a
two-component spinor propagated by a hamiltonian that is also split
into real and imaginary parts.  In quantum algorithms this corresponds
to spending a single additional qubit that codes whether we deal with the real- or
imaginary part of the state \cite{aharonov_simple_2003}.
We formalize the concept by the following definition:
    \begin{definition}\label{def.realif}
    The ``realification'' $\Phi$ of a quantum state $\Psi$ is a map $\Phi:\mathbb
    C^{2^{\Nb}}\rightarrow  \mathbb R^{2^{\Nb+1}}$, $\Psi_i\mapsto
    \Phi_i=(\Re \Psi_i,\Im \Psi_i)\,\forall i=0,\ldots,2^{\Nb}-1$, which obviously constitutes a
    bijective map.  Its inverse mapping $\Phi^{-1}$ is called
    ``complexification''.  
    \end{definition}
The additional grabit (or qubit) will be called ``ReIm grabit'' (or
``ReIm qubit''). $\Nb$ will denote the number of grabits 
    including the ReIm grabit, i.e.~$\Nb=n_\text{bit}+1$, where
    $n_\text{bit}$ is the number of qubits with complex-valued
    amplitudes. 
% The index $i$ of a component of a vector such as $\Psi$ will
% also be denoted as $\bm i$, meaning the binary
% representation of $i$. 
\\

Every pure quantum state $\Psi$ of $\nb$ qubits can be
represented as gradient state, as is made precise by the following
Lemma:
\begin{lemma}
For all $\Psi\in \mathcal H=\mathbb C^{(2^\nb)}$, there exists a
probability distribution $P\in (\mathbb R^+)^{(2^\Nb)}$ over b4v's
such that $\psi=\partial^{\Nb}P=a\Phi$, where $\mathbb R\ni a\ne 0$
and $\Phi\in (\mathbb R)^{(2^\Nb)}$ with $\Nb=\nb+1$ is the
realification of $\Psi$. 
\end{lemma}
\begin{proof}
  The lemma is proven by providing a trivial representation: Set
  $Q_{2\bm i}=\Phi_{\bm i}$ if $\Phi_{\bm i}\ge 0$, $Q_{2\bm i+(0\ldots 01)}=|\Phi_{\bm i}|$ if $\Phi_{\bm i}<0$, and
  all other $Q_{2\bm i+\bm \sigma}=0$. Then ${\tilde\psi}_{\bm
    i}=\partial^{\Nb}\tilde{P}_{\bm i}=\text{sign}(\Phi_{\bm
    i})|\Phi_{\bm i}|=\Phi_{\bm i}$. Now normalize, $P=Q/||Q||_1=Q/\sum_{\bm i}|\Phi_{\bm i}|$. Then $P$ is a
  valid probability distribution and due to linearity
  \begin{equation}
    \label{eq:prlem}
    \psi_{\bm i}=(\partial^{\Nb}P)_{\bm
      i}=\frac{(\partial^{\Nb}Q)_{\bm i}}{\sum_{\bm j}|\Phi_{\bm
        j}|}=\frac{\Phi_{\bm i}}{||\Phi||_1}\,, 
  \end{equation}
  i.e.~$\psi=a \Phi$ with $a>0$.
\end{proof}
Of course, the representation used in the proof is not unique, as the
$\Phi_{\bm i}$ fix only a linear combination of $P_{2\bm i+\bm
  \sigma}$ with alternating signs for different $\bm \sigma$, which
generalizes the gauge choice observed for a single grabit.  The
prefactor $a$ will, in general, depend on the gauge choice and on
the state $\Phi$. \\

Let $U$
be a complex unitary matrix with matrix elements $U_{ij}$ in the
computational basis that propagates a complex quantum state $\Psi$ with
elements $\Psi_j$, $\Psi'_i=U_{ij}\Psi_j$ (with summation convention
over pairs of indices).  Then, coding real- and
imaginary part of $\Psi$ as in definition \ref{def.realif},
$({\Phi}_{i0},{\Phi}_{i1})\equiv (\Re \Psi_i,\Im \Psi_i)$,
the real state $\tilde{\Phi}$ is propagated according to
${\Phi}'_i=\tilde{U}_{ij}{\Phi}_j$ where each 
complex matrix element $U_{ij}$ is replaced by a real 2x2 matrix,
\begin{equation}
  \label{eq:Util}
  U_{ij}\mapsto
  (\tilde{U}_{i\nu,j\mu})_{\nu,\mu\in\{0,1\}}=\begin{pmatrix}
    \Re U_{ij}&-\Im U_{ij}\\
    \Im U_{ij}&\Re U_{ij}
    \end{pmatrix}\,.
\end{equation}
We 
always choose the ReIm qubit as least
significant bit, and numbering of basis states is always in increasing
order of the binary labels. 

\subsubsection{General phase gate}
As simplest example consider the gate that adds a relative phase
$\phi$ to the state $\ket{1}$ of a qubit, while leaving the phase of
$\ket{0}$ untouched, i.e.~
\begin{equation}
  \label{eq:Uphi}
  U_\phi=\begin{pmatrix}
    1 & 0\\
    0& e^{i\phi}
    \end{pmatrix}\mapsto \tilde{U}_\phi=\begin{pmatrix}
    \bm 1_2 & \bm 0_2\\
    \bm 0_2& \begin{matrix}
    \cos \phi & -\sin \phi\\
    \sin\phi& \cos\phi
    \end{matrix}
    \end{pmatrix}\,.
\end{equation}
I.e.~written as a real gate acting on real state vectors, the general phase
gate can be seen as a controlled real phase gate acting as a rotation
in the state space of the second (target-)qubit (=ReIm qubit) if the first
(control-)qubit is in the state $\ket{1}$.  \\

For translation to a grabit gate, consider hence first the action of
the rotation gate on the second qubit.  Assume $0<\phi\le \pi/2$
first, i.e.~$\cos\phi\ge0, \sin\phi>0$. As before, we find the
corresponding stochastic map by studying its action on the
computational basis states, e.g.~$\ket{0}\mapsto
\cos\phi\ket{0}+\sin\phi\ket{1}$, and using the linearity of convex
combinations.  Requesting maximum possible contrast we then have that
the byte4 probability vector is mapped as $(1,0,0,0)^t\mapsto
(q,0,1,0)/(1+q)$, where $q=|\cot(\phi)|$. Altogether we get the
stochastic map
\begin{equation}
  \label{eq:SRphi1}
  S_{R_\phi}^{(1)}=\frac{1}{1+q}\begin{pmatrix}
    q&0&0&1\\
    0&q&1&0\\
    1&0&q&0\\
    0&1&0&q
    \end{pmatrix}\mbox{, } 0<\phi\le \pi/2\,.
  \end{equation}
  The  stochastic process is implemented by acting on the drawn
  realizations: E.g.~if the b4v is 0, we flip it with probability
  $1/(1+q)$ to 2 and keep it with probability $q/(1+q)$, and
  accordingly for the other b4v's. The stochastic process is defined
  even for a
  single realization. 
In the other quadrants, the signs of $\cos\phi$ or $\sin\phi$
change. A sign change of $\cos\phi$ is implemented by moving $q$ in
$S_{R_\phi}^{(1)}$ one row up or down while keeping the same blv.  E.g.~for $\pi/2<\phi\le \pi$
we have 
\begin{equation}
  \label{eq:SRphi2}
  S_{R_\phi}^{(2)}=\frac{1}{1+q}\begin{pmatrix}
    0&q&0&1\\
    q&0&1&0\\
    1&0&0&q\\
    0&1&q&0
    \end{pmatrix}\mbox{, } \pi/2<\phi\le \pi\,,
\end{equation}
and in the remaining quadrants
\begin{equation}
  \label{eq:SRphi34}
  S_{R_\phi}^{(3)}=\frac{1}{1+q}\begin{pmatrix}
    0&q&1&0\\
    q&0&0&1\\
    0&1&0&q\\
    1&0&q&0
  \end{pmatrix}\mbox{, }
   S_{R_\phi}^{(4)}=\frac{1}{1+q}\begin{pmatrix}
    q&0&1&0\\
    0&q&0&1\\
    0&1&q&0\\
    1&0&0&q
    \end{pmatrix}
\end{equation}
for $-\pi<\phi\le -\pi/2$ and $-\pi/2<\phi\le 0$, respectively.\\

The question of normalization becomes relevant for
controlled
gates. 
E.g.~$\ket{\Phi}=\ket{c,t}=(1,0,1,0)^t/\sqrt{2}\stackrel{\tilde{U}_\phi}{\mapsto}(1,0,\cos\phi,\sin\phi)^t/\sqrt{2}$.
I.e.~the 2-norm of the vector in the $c=1$ sub-space $\{\ket{10},\ket{11}\}$ is
conserved, whereas the stochastic map does not conserve the 2-norm,
and hence modifies the ratio of the 2-norms in the $c=1$ vs the $c=0$
subspace. Because
the amplitude of the $c=1$ subspace cannot be increased beyond
what the stochastic matrices $S^{(i)}_{R_\phi}$ deliver, in order to
remedy this we reduce the
amplitude of the $c=0$ subspace instead using the gauge degree of
freedom. For this we make an ansatz for an amplitude reduction gate
acting as stochastic map on the control-grabit $c$, but activated
only if $c=0$. It can hence be restricted to the $\ket{0},-\ket{0}$
subspace, and therefore represented as
\begin{equation}
  \label{eq:A}
  R_2=\begin{pmatrix}
    1-r_0& r_0 \\
    r_0 & 1-r_0\\
    \end{pmatrix}\,.
\end{equation}
% We request that the amplitudes of $0$ and $1$ of the control-grabit
% state are reduced by the same 
% amount. %  i.e.~with $\bm P'=A \bm P$,
% % $(P_0-P_1)/(P_2-P_3)=(P_0'-P_1')/(P_2'-P_3')$. \qdq ??
% This leads to
% $r_0=r_1$.  
The reduction of the 2-norm of
states in the $c=0$ subspace must be by the same amount as of those in the
$c=1$ subspace due to the controlled rotation, which is $||\psi'||_2={\mathcal N}||\psi||_2$
with $\mathcal N=\sqrt{1+q^2}/(1+|q|)$, and gives
% the two-qubit state
% $||\psi'||_2=||\psi||_2=\mathcal N$ leads to 
$r_0=(1-\mathcal N)/2$.
  Altogether the general-phase gate ${U}_\phi$ on a single qubit becomes a
  real gate $\tilde{U}_\phi$ on two qubits that is emulated with a 16
  $\times 16$ stochastic map on two grabits of the form
  \begin{equation}
    \label{eq:Sfull}
    S_{R_\phi}^{\text{full},i}=R_2\otimes \bm 1_4\oplus \bm 1_2\otimes S_{R\phi}^{(i)}\,,
  \end{equation}
  It has the clear
interpretation of reducing the amplitude of grabit components in the
$\pm \ket{0}$ subspace of the first grabit when the conditional rotation of
the second grabit is not activated, while leaving components of the
$\pm 1$ subspace of the first grabit
alone but activating the conditional rotation of the second grabit iff
the first grabit is in state $\pm \ket{1}$.
Since $r_0<1/2$, the sign of the components of the grabit-state is
kept. The amplitude reduction gate \eqref{eq:A} is implemented as a
stochastic process where $0\leftrightarrow 1$ with probability $r_0$,
whereas with probability $1-r_0$ the two b4v's are not flipped, and
correspondingly for $2\leftrightarrow 3$. 
\\ 
If we have more than two grabits (i.e.~more than one complex qubit)
this construction is still sufficient, i.e.~no additional amplitude
reductions need to be introduced: the full Hilbert space decomposes into
the two subspaces with the control grabit either in states $\pm \ket{0}$ or
$\pm \ket{1}$.  If we have controlled unitaries controlled by more
than one control, the construction works accordingly, i.e.~in the
subspace where the unitary is not activated the 2-norm of the components
of the state vector shall be reduced to the same value resulting from the
unitary applied in the complementary subspace.  For this it is enough
to only act with the amplitude reduction gate on the grabits involved
in the control.

\subsubsection{Controlled general phase gate}
The controlled general phase gate on two complex qubits translates
first to a doubly controlled rotation gate on the Re/Im qubit if
$c=t=1$, i.e.~a three-qubit gate. The stochastic map $S_{c-R_\phi}$ on
three grabits 
that emulates it has an $S_{R_\phi}$ block that acts on the Re/Im grabit in the four subspaces
corresponding to $\pm\ket{1}_c\pm\ket{1}_t$ of the control and target grabit, and amplitude reductions in the 
complementary subspaces.  I.e.~we apply the rotation gate to the
ReIm-grabit iff the blv's $c=t=1$, and apply an amplitude
reduction gate to the states $\ket{c,t}$ otherwise.  It is enough to implement
the amplitude reduction on a single grabit, e.g.~for $c=0$ implement
it on $c$, which will reduce the amplitudes of all states
$\ket{c,t}=\ket{\pm 0,\pm 0}$ and $\ket{\pm 0,\pm 1}$. For the state
$\ket{10}$ it can be implemented on $\ket{t}$. Altogether we get for
the controlled general phase gate the three-grabit stochastic map
  \begin{equation}
    \label{eq:Sfull}
    S_{c-R_\phi}^{\text{full},i}=R_2\otimes \bm 1_{16}\oplus R_2\otimes \bm 1_{4}\oplus \bm
    1_2\otimes S_{R_\phi}^{(i)}\oplus R_2\otimes \bm 1_{4}\oplus \bm
    1_2\otimes S_{R_\phi}^{(i)}\,. 
  \end{equation}
Non-trivial stochastic matrices, i.e.~matrices that cannot be reduced to a permutation of b4vs, will be called ``interference-generating stochastic matrices'' in the following, and their set denoted by $\mathcal I$. From the stochastic matrices considered so far, only the Hadamard gate, the general phase gate, and the controlled general phase gate  (for phases $\varphi \ne k \pi$, $k\in \mathbb Z$) are in $\mathcal I$. 

  \subsection{Emulation of a quantum algorithm}
As is well-known, a universal gate set $\mathcal G$ that allows one to approximate any unitary on $\mathbb C^{2^\nb}$ with arbitrary precision is given by
the Hadamard-gate, the $T$-gate (which is the gate \eqref{eq:Uphi} with
$\phi=\pi/4$), and the CNOT gate (see e.g.~\cite{Nielsen00} p.189ff).  
  From the stochastic emulation of the elementary gates in the previous
subsection, we obtain the following theorem:
  \begin{theorem}\label{thm1}
Let $U$ denote a pure-state quantum algorithm, i.e.~a
unitary transformation on the Hilbert space $\mathbb C^{2^\nb}$
composed from a sequence of gates $U\po{i}$
from the universal gate set 
$\mathcal G$, i.e~$U=U\po{\Ng}U\po{\Ng-1}\ldots U\po{1}$.  Let
$\Psi\po{0}$ be the input state, and $\Psi\po{N_g}$ the output state, $\Psi\po{N_g}=U
\Psi\po{0}$. Denote with $\Phi\po{0}$ and $\Phi\po{N_g}$ the realifications of $\Psi\po{0} $ and
$\Psi\po{N_g}$, respectively, in $\mathbb R^{2^\Nb}$ with $\Nb=\nb+1$.  Let $P\po{0}$ be
a probability distribution over b4vs of $\Nb$  grabits, such
that $\psi\po{0}=\partial^{\Nb}P\po{0}=a\po{0}\Phi\po{0}$ 
with $a\po{i}\in\mathbb R\setminus\{0\}$.  
Then $\psi\po{N_g}\equiv\partial^{\Nb}P\po{N_g}$ with $P\po{N_g}=S\,P\po{0}$ and $S=S^{(N_g)}S^{(N_g-1)}\ldots S^{(1)}$, where $S^{(i)}$ are the stochastic matrices corresponding to gates $U^{(i)}$ (cf.~eqs.(\ref{eq:SH},\ref{eq:Sfull}) and $S_\text{CNOT}$ described in Sec.\ref{Sec.CNOT}), satisfies $\psi\po{N_g}=a\po{N_g}\Phi\po{N_g}$ with $a\po{N_g}\in\mathbb R\setminus\{0\}$. 
  \end{theorem}
  \begin{proof}
For all these gates, we have found corresponding stochastic maps that
we can now iterate in the sequence defined by the quantum
algorithm. I.e.~for an arbitrary $U\po{i}\in\mathcal G$ with 
realification $\tilde{U}^{(i)}$ there is a corresponding $S\po{i}$ such
that with $\psi\po{i-1}=a\po{i-1}\Phi\po{i-1}=\partial^{\Nb}P\po{i-1}$
(and $a\po{i-1}\ne 0$)
and $P\po{i}=S\po{i}P\po{i-1}$, we have 
$\psi\po{i}=\partial^{\Nb}P\po{i}=a\po{i}\tilde{U}\po{i}\Phi\po{i-1}=(a\po{i}/a\po{i-1})\tilde{U}\po{i}\psi\po{i-1}$,
$i=1,\ldots,\Ng$. Iterating this, we get
\begin{eqnarray}
  \label{eq:iterate}
  \psi\po{\Ng}&=&\partial^\Nb
                  P\po{\Ng}=a\po{\Ng}\Phi\po{\Ng}=a\po{\Ng}\tilde{U}\po{\Ng}\Phi\po{\Ng-1}\nonumber\\
  &=&\frac{a\po{\Ng}}{a\po{\Ng-1}}\tilde{U}\po{\Ng}\psi\po{\Ng-1}\nonumber\\
  &=&\frac{a\po{\Ng}}{a\po{\Ng-1}}\frac{a\po{\Ng-1}}{a\po{\Ng-2}}\tilde{U}\po{\Ng}\tilde{U}\po{\Ng-1}\psi\po{\Ng-2}\nonumber \\
  &=&\ldots=\frac{a\po{\Ng}}{a\po{0}}\tilde{U}\po{\Ng}\tilde{U}\po{\Ng-1}\ldots\tilde{U}\po{1}\psi\po{0}\nonumber % \\
  % &=&{a\po{\Ng}}\tilde{U}\po{\Ng}\tilde{U}\po{\Ng-1}\ldots\tilde{U}\po{1}\Phi\po{0}
      \,,
\end{eqnarray}
i.e.~up to an overall prefactor, the grabit state $\psi\po{0}$ is propagated with exactly the same quantum algorithm $\tilde U=\tilde{U}\po{\Ng}\tilde{U}\po{\Ng-1}\ldots\tilde{U}\po{1}$ as the realification $\Phi\po{0}$: $\Phi\po{Ng}=\tilde{U}\po{\Ng}\tilde{U}\po{\Ng-1}\ldots\tilde{U}\po{1}\Phi\po{0}$.  At the same time we have $P\po{\Ng}=S\po{\Ng}\ldots S\po{1}P\po{0}$, i.e.~the 
stochastic map can indeed be composed from the stochastic maps
corresponding to the universal gates discussed before. 
  \end{proof}

\subsection{Demonstration with two simple quantum algorithms}
  We first illustrate the basic functioning of the gradient-state approach at the
example of two quantum algorithms, before coming back to more advanced ones in section \ref{Sec.qalgo}.

\subsubsection{Deutsch-Josza algorithm}
The Deutsch-Josza algorithm (DJA) is maybe the simplest quantum algorithm
that demonstrates a quantum advantage over classical algorithms.  In
addition, one can clearly point to the origin of the advantage, namely
many-body interference.  The algorithm allows one to decide whether a
Boolean function $f:\{0,1\}^{\otimes \Nb}\rightarrow \{0,1\}$ is constant or
balanced, assuming that it has only these two options.  Here, balanced
means that $f=1$ for as many inputs $x$ as for which 
$f=0$. Functions can be evaluated reversibly by keeping the inputs,
and making the map bijective, via $f: x,y\mapsto x,y\oplus f(x)$.
Classically, one has to evaluate $f$ at least once and at the most
$2^{\Nb-1}+1$ times in order to decide, whereas the algorithm by
Deutsch-Josza succeeds with probability one with a single reversible
evaluation of $f$, applied to a superposition of all computational
basis states. The algorithm has 5 steps: 0. Preparation of input states
$\ket{0}^{\otimes \Nb}\otimes\ket{1}$. 1. Application of Hadamard gates
on all qubits. 2. Reversible calculation of $f$, with the first $\Nb$
qubits untouched, the last one receiving $y\oplus
f(x)$. 3. Application of Hadamard gates on the first $\Nb$ qubits and
4. their measurement in the computational basis. \\

We illustrate the algorithm for the simplest
case of two grabits and $f(x)=x$, i.e.~$f$ is balanced. The
probability distribution over the 16 possible byte4-values after the
first two
steps above % are $\bm P_0=(1,0,0,0)\otimes(0,0,1,0)$, $\bm
% p_1=(1/2,0,1/2,0)\otimes(1/2,0,0,1/2)$. In terms of probability
% distributions over byte4 values
 which we denote with $\ket{I}\rangle$
 for each grabit, and $\ket{IJ}\rangle$ for two grabits, %we can write
 are
\begin{eqnarray}
  \label{eq:p1}
\ket{P\po{0}}\rangle&=&\kett{02}\,,\\  
\ket{P\po{1}}\rangle&=&(\ket{00}\rangle+\ket{03}\rangle+\ket{20}\rangle+\ket{23}\rangle)/4\,.
\end{eqnarray}
The function evaluation $f(x)$ permutes byte4 values with the
condition of {\it i.)} giving the correct logical values and {\it
  ii.)} preserving the sign of the state.  The concrete example
$f(x)=x$ translates to the logical map $x,y\mapsto x,y\oplus x$.  So
e.g.~logical 10 is mapped to 11, which means that byte4 values 20 are 
mapped to 22, and 23 to 21. Thus, after the function evaluation, we have   
\begin{equation}
  \label{eq:p2}
\ket{P\po{2}}\rangle=(\ket{00}\rangle+\ket{03}\rangle+\ket{22}\rangle+\ket{21}\rangle)/4\,.
\end{equation}
After the last Hadamard gate on the first grabit, we find
\begin{equation}
  \label{eq:p2}
  \ket{P\po{3}}\rangle=(\ket{00}\rangle+\ket{20}\rangle+\ket{03}\rangle+\ket{23}\rangle
+\ket{32}\rangle+\ket{02}\rangle+\ket{31}\rangle+\ket{01}\rangle
  )/8\,, 
\end{equation}
which translates to a two-grabit  wavefunction
\begin{eqnarray}
  \label{eq:psi3}
  \ket{\psi\po{3}}&=&(\ket{00}+\ket{10}+\ket{0}(-\ket{1})+\ket{1}(-\ket{1})+(-\ket{1})\ket{1}+\ket{0}\ket{1}+(-\ket{1})(-\ket{0})+\ket{0}(-\ket{0})/8\nonumber\\
              &=&(0\ket{00}+0\ket{01}+2\ket{10}-2\ket{11})/8\nonumber\\
                  &=&(1/2)\ket{1}\ket{-}\,,
\end{eqnarray}
with $\ket{-}=(\ket{0}-\ket{1})/2$.
So indeed the state is propagated correctly, and we can extract $\psi\po{3}$
from the final probability distribution. For numerical implementation the ReIm grabit is not
needed, as the state remains real in the computational basis at all
times, i.e.~we can set $\Nb=\nb$ here.
A measurement that gives
logical values distributed according to a Born rule (be it Born-1 rule
or the usual Born-2 rule) would lead to a measurement outcome ``1''
with probability one for the first grabit.  In this sense, the
stochastic emulation gives the answer ``$f(x)$ is balanced'' with a
single application of the gate that evaluates the function $f$.
However, in the numerical implementation, $f(x)$ has to be evaluated
for each realization (i.e.~$\Nl$ many times!), and, furthermore, $\Nl$
has to be large in order to have near-perfect cancellation of the
amplitude on logical state $\ket{0}$ of the first grabit.  Secondly,
the physical probability distribution at the end
of the algorithm is a uniform distribution over all logical states,
\begin{equation}
  \label{eq:p3t}
  \kett{\tilde{p}_3}=(\kett{00}+\kett{01}+\kett{10}+\kett{11})/4\,.
\end{equation}
These issues will be considered in section \ref{Sec.Refresh}, but
before we make them more quantitative with a numerical emulation of
the Bernstein-Vazirani algorithm. 
% The latter issue can be remedied by a global refresh at the end of the
% algorithm (note that this is the first stage where a refreshment is
% useful as all previous states are interference-free). In the limit of
% $\Nl\to\infty$ it gives correctly
% $\kett{\tilde{p}_4}=(\kett{20}+\kett{23})/2$, i.e.~the first grabit
% will then be always found in logigal state 1. However, the problem of
% needing many realizations is not solved by the refreshment, as it
% requires estimating the
% state from a histogram which works only for a large number of
% realizations.  

\subsubsection{Bernstein-Vazirani algorithm}
The Bernstein-Vazirani algorithm can be considered a generalization of
the Deutsch-Josza algorithm.  It allows one to determine a linear
function $f(x)=a\cdot x \mod 2$ in a single shot, where $a$ is a
binary string on $\nb$ qubits.   The quantum circuit is exactly as for
the Deutsch-Josza algorithm. At the end of the algorithm, the first
register of $\nb$ qubits contains $a$ (indeed the shown example
of the DJA corresponds to $a=1$, and we saw that the first qubit is in
state $\ket{1}$ at the end of the algorithm).  One might still apply a
final Hadamard on the last 
grabit, in which case we expect a single peak in the final state,
$\ket{\psi}_3=\ket{a}\ket{1}$. Fig.\ref{fig:BV} shows a stochastic
emulation with $\nb=3, \Nl=10000$ for $a=01_2$ (the subscript $_2$
indicates binary representation). We hence
expect a peak 
of $\ket{\psi}_3$ at logical values $011_2=3$, which is found
indeed. The figure also show an analysis of the error-probability as function of
$\Nl$ for different values of $\nb$, based on identifying the
logical value $i$ with the largest absolute amplitude $|\psi_i|$ at
the end of the algorithm as output $a'$ of the algorithm. 500 different
realizations of the algorithm were run with random initial 
$a$ at given $\nb$ (the $\nb$ indicated in the legend refers to the
total number of grabits, including the last one initialized in
$\ket{1}$, so the largest integer $a$ for $\nb=7$ is 64), and the
error probability estimated as relative frequency over all
realizations of the algorithm and initial values of $a$ where in the
end $a'\ne a$.  We see that
the probability 
that this decision rule leads to the wrong value of $a$ decreases roughly exponentially with 
the number of realizations once $\Nl$ is sufficiently large. 
\begin{figure}[t]
  \includegraphics[width=0.4\textwidth]{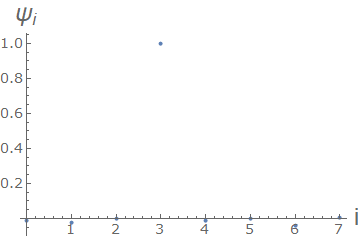}
 \includegraphics[width=0.4\textwidth]{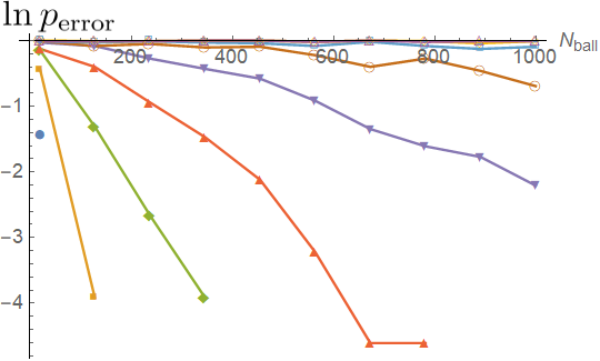}
  \caption{Stochastic emulation of Bernstein-Vazirani algorithm. {\em
      Left:} Estimate $\hat\psi_i$ of final state $\ket{\psi^{(3)}}$ for $\nb=3$ (without using a    ReIm grabit),
    $\Nl=10000$, $a=01_2=1$, including a final Hadamard gate on the
    third grabit initialized in $1$, i.e.~the expected output peak of
    the 3-grabit $\ket{\psi}$ is at $011_2=3$. %\qdq{\em Middle:}
       {\em Right:} Error
    probability as estimated from 100 random values
    of $a$ as function of $\Nl$ for    $\nb=2,\ldots, 10$ from bottom to top in a lin-log plot. Missing values correspond to  vanishing error of the estimate (=relative frequency); lines are guides to the eye.
          \label{fig:BV} }
\end{figure}

\section{Refreshments}\label{Sec.Refresh}
The examples of the Deutsch-Josza and the Bernstein-Vazirani quantum
algorithm illustrate that while the propagation of the multi-grabit wave
function works correctly, Theorem \ref{thm1} is not enough for
obtaining a useful stochastic algorithm from a quantum algorithm. 
Indeed, extracting the final state as $\psi=\partial^{\Nb}P$ from the
final probability distribution and finding its maximum is clearly not
efficient when the number of qubits 
becomes large. Rather we want that the flow of physical probability
follows the quantum mechanical wavefunction and leads to a peak on
the correct result in the end, such that the correct result is found
with high probability, just as with a quantum computer.  Already at the
example of the two-grabit Deutsch-algorithm we can see, however, that the
physical probability $\tilde{p}_1$ for finding the first grabit in state 1 is not
1. In fact, the final physical probability distribution is evenly
distributed over all computational basis states.\\

This issue can be seen most clearly already at the
example of the simplest non-trivial quantum algorithm showing
interference: $U=H^2=\bm 1_2$. Starting with $\ket{0}$, the final
state is $\ket{0}$ as well. The byte4-value probability distribution of
the grabit-emulation of this circuit is easily found to be $\bm
P=(1/2,0,1/4,1/4)^t$, from which results the final grabit-state
$\ket{\psi}=\ket{0}/2$, i.e.~up to normalization the correct
quantum state is perfectly reproduced. But the physical probability
distribution is $\tilde{\bm p}=(1,1)/2$, i.e.~there is no contrast at all.

This can be generalized
to $n$ iterations of $H^2$ and reveals that there is another
issue: With $S_H$ from \eqref{eq:SH} we find that the
initial byte4-probability distribution $\bm P_0=(1,0,0,0)$
corresponding to initial state $\ket{0}$ is propagated to
\begin{equation}
\label{eq:pH2n}
  \bm P^{H^{2n}}=\frac{1}{4}
  \begin{pmatrix}1\\1\\1\\1
  \end{pmatrix}+\frac{1}{2^{n+1}}
  \begin{pmatrix}1\\-1\\0\\0
  \end{pmatrix}\,,
\end{equation}
which  translates to a grabit-state
$\ket{\psi^{H^{2n}}}=\frac{1}{2^{n}}\ket{0}$. So while the correct
quantum state is reproduced up to normalization, its amplitude
decreases exponentially with the number of Hadamard-gates. If the
estimation of the quantum state is based on a finite number of
realizations, this means
that the number of realizations that effectively contribute to
$\ket{\psi^{H^{2n}}}$ decreases exponentially, and hence the
statistical error, measured as the deviation in 2-norm of the
multi-grabit state estimated from the histogram of realized byte4 values (and normalized to 2-norm equals one) from the  
correct quantum mechanical state, 
increases exponentially with the number of gates $n_\text{gate}$.
This is confirmed numerically in Fig.\ref{fig:expinc}a. In Appendix \ref{app:Destr} we show that the reduction of the effectively contributing samples is in fact a measure of destructive interference and that for all  quantum gates in the universal set considered here and all input states, a reduction by a factor $1/2$ is the worst case.\\

These issues will be addressed now.  A partial solution 
consists in using the gauge-degree of freedom in the parameterization
of the grabit-states to enforce after each interference-generating gate the Born-1 rule, a
process that will be called ``refreshing''.  

% Let us
% first study this in detail for a single grabit, and then for the
% multi-grabit state.
\subsection{General strategy}
We are looking for a process that maps $\bm P$ to $\bm P'$ (implying a
map $\psi\mapsto \psi'$) such that the following two conditions
(called ``refreshment conditions'') are satisfied:
\begin{enumerate}
\item
  \begin{equation}
    \label{eq:1.3}
\psi'=d\cdot \psi\,,    
  \end{equation}
 i.e.~up to a global factor $d\ne 0$ the
  multi-grabit state remains unchanged.  This is crucial for being
  able to continue propagating the correct multi-grabit state by the
  quantum algorithm;
\item and
  \begin{equation}
    \label{eq:1.4}
\tilde p'=c|\psi| =(c/|d|) |\psi'|\,,
  \end{equation}
with $c>0$, which
  ensures that a peaked 
  $\psi$ at the end of the quantum algorithm implies a peaked physical
  probability distribution $\tilde p'$. 
\end{enumerate}
As mentioned before, from a quantum mechanics perspective, one would
prefer $\tilde p'=|\psi'|^2$, but enforcing the Born-1 rule turns out
to be simpler and is sufficient for having an emulation of the
quantum algorithm that succeeds with high probability if the version
with the Born-2 rule does. \\
% While both the amplitudes $\psi_{\bm i}$ and the physical
% probabilities $\tilde{p}_{\bm i}$ are linear in $P_{2\bm i +\bm
%   \sigma}$, the absolute value or any other function that maps $\pm
% \psi_{\bm i}$ to $\tilde{p}_{\bm i}>0$ with the refreshment conditions
% satisfied makes the required mapping $\bm P\mapsto \bm P'$ necessarily
% non-linear. Using the Born-1 rule helps, however, in allowing a
% process that is piece-wise linear.  
In terms of the number of free
parameters, we have $4^{N_\text{bit}}-1$ free byte4 probabilities
$P_{I}$.  Eqs.(\ref{eq:1.3},\ref{eq:1.4}) are $2^\Nb$ equations each
for the amplitudes, but also introduce two new real variables
$c,d$. So all together there are $4^\Nb-2^{\Nb+1}+1$ free
variables. Summing \eqref{eq:1.4} over all $i$, we find
$c=1/\bar\psi$, with 
$\bar\psi\equiv\sum_{i}|\psi_{i}|$. % , i.e.~the
% sum of the absolute values of the old multi-grabit state over all
% computational basis states.
For $\Nb=1$ there is only one free
variable.

\subsection{Refreshing a single grabit}\label{Sec.grabitrefresh}
In the case of a single grabit, eqs.(\ref{eq:1.3},\ref{eq:1.4}) can be
solved immediately. First consider the
case $\psi_i\ge 0$ ($i=0,1$). Adding and subtracting
eq.(\ref{eq:1.3}) for $i=0,1$ then leads to
\begin{eqnarray}
  \label{eq:pfs}
  P'_{2i}&=&\frac{1}{2}(\frac{1}{\bar \psi}+d)\psi_i,\\
P'_{2i+1}&=&\frac{1}{2}(\frac{1}{\bar \psi}-d)\psi_i\,,                
\end{eqnarray}
where $d$ is the remaining free parameter.  Positivity of the
probabilities implies a range for $d$, 
\begin{equation}
  \label{eq:maxd}
  \max(-\frac{1}{\bar
    \psi},-\left(\frac{2}{\psi_i}-\frac{1}{\bar\psi}\right))\le d \le \min(\frac{1}{\bar
    \psi},\left(\frac{2}{\psi_i}-\frac{1}{\bar\psi}\right))\,.
\end{equation}
Hence, $d_\text{max}=1/\bar\psi=1/(\psi_0+\psi_1)$. In the case of a
$\psi_i<0$, it is easily checked that we can just replace $\psi_i\to
|\psi_i|$ everywhere in the calculation of $d_\text{max}$, which leads
to the general result $d_\text{max}=1/\bar\psi$, and 
\begin{equation}
  \label{eq:pti}
  \tilde{p}_i'=\frac{|\psi_i|}{|\psi_0|+|\psi_1|}=\frac{|P_{2i}-P_{2i+1}|}{|P_0-P_1|+|P_2-P_3|}\,.
\end{equation}
Maximizing $d$ under the constraint of positivity of the byte4 probabilities
corresponds exactly to the gauge choice of maximum contrast. We see
from \eqref{eq:pti} that for any $\psi_i\ne 0$ this corresponds to having always either
$P_{2i}=0$ or  $P_{2i+1}=0$.  This has the important physical
interpretation that {\em refreshing means to make the state
``interference-free''}, i.e.~the amplitude corresponds not
to the difference of two byte4 probabilities anymore, but just to one
of them, the other being zero.  This has the desired consequence that
the physical 
probability is equal to the probability obtained from the Born-1 rule.\\

{For the numerical implementation},
it suffices to refresh after each gate $\in \mathcal I$, i.e.~each gate that can potentially generate interference. We analyzed two different versions of a refreshment gate. Both of them need the estimation of the grabit state before refreshment and base the subsequent transformations on it, which renders refreshments non-linear in $\bm P$.  The non-linearity in $\bm P$  is also implied by the fact that refreshments map many different b4v probability distributions to the same interference-free b4v probability distribution. It is a step that reduces entropy. 

A first implementation works in three steps:  
\begin{enumerate}
\item Estimate the state from the histogram $R$ of realized b4v's, i.e.~evaluate $\hat{\psi}_{i}$ from \eqref{eq:psiest}.  As the individual $R_I$ converge for $\Nl\to\infty$ to the $P_I$ that determine their distribution, i.e.~$R_I\stackrel{\Nl\to\infty}{\longrightarrow}P_I$, also $\hat{\psi}\stackrel{\Nl\to\infty}{\longrightarrow} \psi$. 
\item For each logical value $i$, determine $\text{sign}
  (\hat\psi_i)$. A b4v $I=2i+1$ is  called majority value (and $I=2i$ minority value) 
 if $R_{2i+1}>R_{2i}$, corresponding to 
  $\text{sign} (\hat\psi_i)<0$.  For $\text{sign} (\hat\psi_i)\ge 0$, $I=2i$ is the majority value and $I=2i+1$ the 
  minority value. Move all realizations
  corresponding to minority values to the corresponding majority
  values, a step called ``minority relocation'' (MIRE).  MIRE
  changes the probability distribution over logical 
  values $\tilde{p}_i\mapsto {\tilde{p}^M}_i$,  and the state $\psi$ to
  ${\psi}^M$ but enforces
  $\tilde{p}^M_i=|{\psi^M}_i|$.
  \item Move realizations between the remaining byte4 values $I=2i+1$
    (or $I=2i$) until the original estimated ratios between probability
    amplitudes are restored as closely as possible, i.e.~$\hat\psi^M_i\mapsto \psi'_i\simeq d \,\hat\psi_i$ 
    (a step called ROAR=''restoration of original amplitude ratios'').  
  \end{enumerate}
Note that a perfect restoration of
    the original amplitude ratios may not always be possible, as the
    histograms lead to integer-ratios for the estimates of $\psi_{i}$. E.g.~for $\Nl=11$ and a histogram
$\bm R=(4,0,4,3)^t/11$, we estimate $\ket{\psi}\propto (4,1)$.  After
MIRE, $\bm R^\text{M}=(4,0,7,0)^t/11$.  With ROAR, the closest we can get to the estimated
$\ket{\psi}$ is $\bm R^\text{M}=(9,0,2,0)^t/11$, giving $\ket{\psi}\propto
(4.5,1)$. However, this type of error is of order $1/\Nl$ and hence becomes negligible for $\Nl\to\infty$.
  The  stochastic map that implements this first type of refreshments will be denoted by $Rf_1$.  

In Fig.\ref{fig:expinc} we see at the example of $H^{n_H}$ on a single qubit that when the refreshment gate is
implemented after each Hadamard gate, the error increases only  
$\propto \sqrt{n_H}$ compared to the
exponential increase without the refreshment gates. As the Hadamard gates keep the state real at all times, we can compare with $\Psi$ directly and do not need to consider the realification $\Phi$. The decrease of the overall error compared to exact unitary propagation is $\propto
1/\sqrt{\Nl}$, for both using the refreshment gates or not.  Both 
behaviors can be understood analytically, see  Appendix \ref{app:A}.

\begin{figure}[t]
  \includegraphics[width=0.45\textwidth]{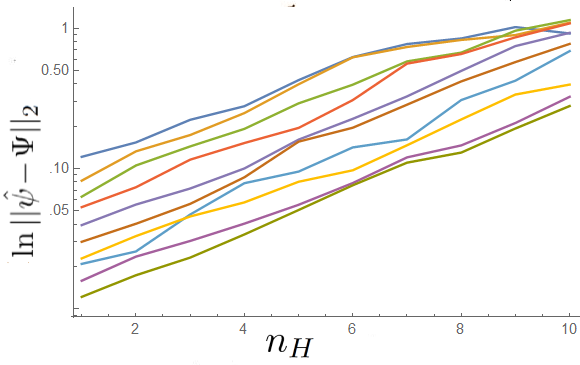}
  \includegraphics[width=0.45\textwidth]{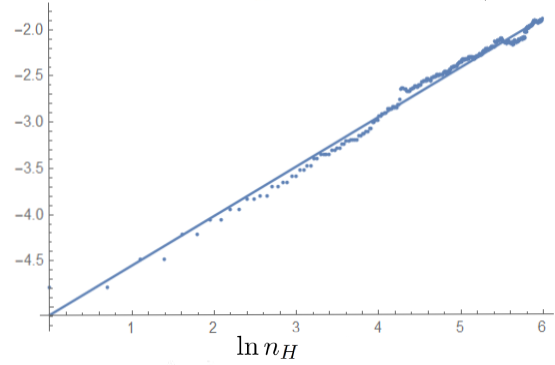}
  \caption{Error $||\hat\psi-\Psi||_2$ of the final state estimate compared to the exact quantum state $\Psi$ after a
    sequence of Hadamard gates $H^{n}$ on the first grabit  as
    function of the number $n_H$ of Hadmards for $\Nb=2$ grabits,
    starting in the initial state $\ket{00}$. $\hat\psi$ is renormalized to $||\hat\psi||_2=1$ for comparison with $\Psi$.  The second (ReIm-)grabit
    remains in $\ket{0}$.   
    {\em
      Left:}  Without refreshment gates the error increases exponentially
    with the number of gates. The different curves are averaged over
    $\Nr=50$ realizations and correspond to
    different values of $\Nl$ ranging from $\Nl=50$ to $\Nl=5000$ in
    10 linear steps (from top to bottom at $n_H=1$). 
   {\em Right:} Same as on the left but with an additional refreshment
   gate after each Hadamard gate for a sequence of up to 200 Hadamard
   gates on the first grabit in a log-log plot,
   i.e.~$\ln||\hat{\psi}-\Psi||_2$ as function of $\ln n_H$,
   where $n_H$ now counts the total number of gates including the
   refreshment gates. Emulation with $\Nl=10000$ and $\Nr=100$
   realizations.  The fit gives a power law
   $\ln||\hat{\psi}-\Psi||_2\simeq -5.08413 + 0.532838 \ln n$,
   i.e.~the error increase roughly as the square root of the number of
   gates.
   %Fig. left: From RefreshGlobal_20190528.nb, 2 bit - Error as
   %function of Nball and gate number for sequence of Hadamards 
   %Fig. right: From AccumulatedErrors_20190701.nb, 2 bit - Error as
           %function of NGate for sequence of Hadamards  followed by
           %RefreshGates; NH=200, Nball=10000, Nreal=100    
          \label{fig:expinc} }
\end{figure}

\subsection{Multi-grabit Refreshments}
Since entangled many-grabit states are based on byte4-value
correlations between different grabits it is clear that 
refreshing grabits independently is not sufficient.
In fact, if correlations exist it is counterproductive, as the example of the quantum circuit in Fig.\ref{fig:circ} shows.
The byte-4 probability
distributions $\ket{P_n}\rangle$ after step $n$ (with the first step
the first Hadamard gate) are
$\ket{P_1}\rangle=(1/2)(\kett{0}+\kett{2})\kett{0}$,
$\ket{P_2}\rangle=(1/2)(\kett{00}+\kett{22})$,
$\kett{P_3}=(1/4)(\kett{00}+\kett{20}+\kett{02}+\kett{32})$.  To these
correspond physical probability distributions
$\kett{\tilde{p}_1}=(1/2)\kett{00}+\kett{10}$,
$\kett{\tilde{p}_2}=(1/2)\kett{00}+\kett{11}$, 
$\kett{\tilde{p}_3}=(1/4)\kett{00}+\kett{10}+\kett{01}+\kett{11}$, and
grabit states
$\ket{\psi_1}=(1/2)(\ket{00}+\ket{10})$,
$\ket{\psi_2}=(1/2)(\ket{00}+\ket{11})$,
$\ket{\psi_3}=(1/4)(\ket{00}+\ket{01}+\ket{10}-\ket{11})$.  We see
that at each stage the Born-1 rule is satisfied,
i.e.~$\brakett{i|\tilde{p}_n}=|\braket{i|\psi_n}|$. Compared to the simple
sequence $H^2$, the CNOT between the Hadamard
gates prevents interference, so the final state is interference-free
and the Born-1 rule automatically fulfilled.  Hence, there is no need
for a refreshment gate. The reduced byte-4 probability distribution of
the first grabit alone, however, looks as if the CNOT was not there,
i.e.~gives $\kett{P^{(1)}_3}=(1/2)\kett{0}+(1/4)\kett{2}+(1/4)\kett{3}$. If
a refreshment of the first grabit is  performed, it ends up in
$\kett{P^{(1)}_4}=\kett{0}$.  But since the correlations are kept during the
refreshment process, this means that $\kett{20}\mapsto\kett{00}$ and
$\kett{32}\mapsto\kett{02}$. So the physical probability distribution
of the two grabits becomes $\kett{\tilde{p}_4}=(1/2)(\kett{00}+\kett{01})$, and
the two-grabit wavefunction
$\ket{\psi_4}=(1/2)(\ket{00}+\ket{01})$. So while the Born-1 rule is
satisfied, we end up with the wrong state and the wrong physical probability
distribution. \\
The above example clearly shows that when correlations are present,
single-grabit refreshments are inappropriate. However, the algorithm
presented for single-grabit refreshments can be easily generalized to
multi-grabit refreshments, by determining majority and minority
contributions to each byte logical value of all grabits combined, rather
than for each grabit separately.  Since all grabit states in the
example given are interference-free, the refreshment algorithm will
terminate immediately without doing anything in that case.  In
general, many different byte-4 values will contribute to the majority-
and minority-components of a multi-grabit byte logical value.
Indeed, eq.\eqref{eq:psiN} can be written as   
 \begin{equation}
  \label{eq:psiN2}
  \hat\psi_{\bm i}=\sum_{\bm \sigma, \pi_\sigma=0} R_{2\bm i+\bm\sigma}-\sum_{\bm \sigma, \pi_\sigma=1} R_{2\bm i+\bm \sigma}\,, 
\end{equation}
where $\pi_{\bm\sigma}=\sum_{i=1}^\Nb \sigma_i \,\text{mod\,}2$ is the parity of a given
gradient-value $\bm\sigma$. The majority contributions to the estimate
$\hat\psi_{\bm i}$ are the ones with a parity that determines the overall
sign of $\hat\psi_{\bm i}$. The fact that different majority and minority
components can contribute to the same $\hat\psi_{\bm i}$ is owed to the
possibility of shifting signs between different factors in a
multi-grabit state.  The generalization of $Rf_1$ in
Sec.~\ref{Sec.grabitrefresh} is the following: 
\begin{enumerate}
\item Estimate the multi-grabit state from \eqref{eq:psiest}, yielding $\hat\psi_{\bm i}$. 
\item For each logical value $\bm i$, determine $\text{sign}
  (\hat\psi_{\bm i})$, which in turn determines what are the majority and
  minority contributions to $\hat\psi_{\bm i}$.  In MIRE move all realizations
  corresponding to minority values to gradient-value $\bm\sigma=\bm 0$ if the majority
  parity is even, or $\bm\sigma=(0,\ldots,0,1)$ if the majority parity
  is odd, i.e.~one chooses a single canonical majority byte4 value for each $I$ by always locating the sign in the same grabit (a step we call ``sign concentration'').
This
  changes the probability distribution over logical 
  values, $\tilde{p}_i\mapsto {\tilde{p}^\text{M}}_i$ and $\psi_i\mapsto {\psi}_i^\text{M}$, modifying in general the amplitude ratios, but keeps the signs of the $\hat{\psi}_{\bm
    i}$ and enforces   $\hat{\tilde{p}}^\text{M}_i=|\hat{\psi_i}^\text{M}|$. 
  \item Move realizations between the remaining byte4 values until the original ratios between estimated probability
    amplitudes are restored as closely  as possible, i.e.~$\hat{\psi}_i^\text{M}\mapsto \hat\psi'_i\simeq d \,\hat\psi_i$ (ROAR).  
  \end{enumerate}

  The whole
algorithm can be implemented rather efficiently with an effort
$\propto \Nl$, 
by working on the realized support of the histogram
$\bm R$ only, and keeping a table of realizations  after MIRE
corresponding to 
majority and minority byte-4 values for each $\bm i$, as well as a
deviation table $\delta_{\bm i}$ of the new estimate $\hat\psi$ from the stored one before
MIRE (and both normalized to 1-norm=1). In ROAR one then
flips successively realizations of $\bm i$ with $\delta_{\bm i}>0$ to ones of $\bm j$ with $\delta_{\bm j}<0$, updating both the lists of 
realizations and the deviation table (which only needs a counting up
and down of the histogram  values, no re-evaluation based on all
realizations), till for the involved $\bm i$ or $\bm j$ the sign of
the deviation-table entry changes, or no more realizations to be
flipped are available for the given $\bm i$  or $\bm j$, in which case
one moves on to the next $\bm i$ or $\bm j$. Alternatively, one can implement ROAR with a Monte-Carlo algorithm based on random maps $I\mapsto J$ with a target function $||\hat{\psi}^M-\hat{\psi}||_2$ to be minimized.\\

\begin{figure}[h]
    \centering
    \includegraphics[width=0.5\textwidth]{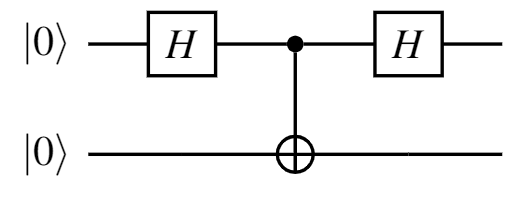}
    \caption{Exemplary quantum circuit to demonstrate the need for multi-grabit refreshments.}
    \label{fig:circ}
\end{figure}

An alternative refreshment algorithm is this one ($Rf_2$, ``removal of socket''): 
\begin{enumerate}
\item Remove all pairs of b4vs that differ only by their gradient
  value. This removes the ``socket'' of realizations that do not
  contribute to the $\hat\psi_i$ and reduces $\Nl$ to a new value $\Nl^{RS}$.  
\item Replicate all remaining b4vs an integer number of times, such that the new total number $\Nl'$ of realizations is mapped to $\Nl'=k \Nl^{RS}\ge \Nl$ with $k\in \mathbb{N}$ such that  $\Nl'$ is the smallest such integer $\ge \Nl$.\\
\end{enumerate}

While $Rf_2$ appears to be simpler than $Rf_1$, it
leads to a diffusive process of $\Nl'$, which requires the allocation
of sufficient memory such that values $\Nl'>\Nl$ can be
accommodated, or the reinitialization of the memory array containing
all b4vs after the refreshment. However, $\Nl'\le 2\Nl$ at all times,
due to the fact that the second step is executed only if $\Nl^{RS}$ drops
below the initial $\Nl$.   At first sight, the 2nd implementation
also seems to avoid estimation of $\psi$, but an equivalent
implementation is given by simply recreating a new set of 
$k \Nl^{RS}$
realizations based on $\hat\psi$, by replicating each set of existing
$\bm i$ $k$ times, where due to $\hat\psi_i\in\mathbb Q$ we can choose
the smallest $k$ such that $k\hat\psi_i\in\mathbb N$ $\forall i$, and
$\Nl'\ge \Nl$. A practical modification of $Rf_2$ (called $Rf_3$) consists in
allocating a fixed memory space for $2\Nl$ realizations. After
removing the socket and calculating the estimate $\hat{\psi}$, we
create $\Nl'$ realizations by assigning $ \left\lfloor 2\hat{\psi}_j
  \Nl \right\rfloor$ realizations to $j$ and then distribute all
remaining samples $N_{\text{ball,rest}} = 2\Nl - \Nl'$ onto those $j$
that have the highest fractional value $2\hat{\psi}_j \Nl
-\left\lfloor 2\hat{\psi}_j \Nl \right\rfloor$ remaining, thereby
minimizing the deviation $||\hat{\psi}-\hat{\psi}'||_1$  between the
state before refreshment $\hat{\psi}$ and the state after refreshment
$\hat{\psi}'$ while using all the available memory space to represent
the grabit state.

\subsection{Efficiency of Refreshments}
The following theorem shows that the grabit approach can perfectly
emulate a quantum algorithm via a stochastic classical algorithm in
the sense that for $\Nl\to\infty$ and refreshments implemented, the final
physical probability distribution is the same as the quantum
mechanical one, up to the different Born rule:
\begin{theorem}\label{thm2}
  Let $U$ be a unitary quantum algorithm of finite depth $N_g$, $\Phi\po{0}$, $\Phi\po{N_g}$ the realifications of the states $\Psi\po{0} $, and
$\Psi\po{N_g}$ as in Theorem \ref{thm1}, and
$S=S^{(N_g)}S^{(N_g-1)}\ldots S^{(1)}$, the corresponding stochastic
map on byte-4 values that emulates $U$. Let
${S}^{(R,N_g)}=R^{(N_g)}S^{(N_g)}R^{(N_g-1)}S^{(N_g-1)}\ldots
R^{(1)}S^{(1)}$ be a stochastic map on byte-4 values with each
stochastic matrix $S^{(j)}\in \mathcal I$ followed by an $\Nb$-grabit
refreshment gate, i.e.~$R^{(j)}=Rf_1$ for all $j$ with $S^{(j)}\in
\mathcal I$ (alternatively: $R^{(j)}=Rf_2$ for all $j$ with
$S^{(j)}\in \mathcal I$), and $R^{(j)}=\text{Id}$ else, where Id is
the identity operation on all byte-4 values.  Let $P^{(R,N_g)}_I$ be
the byte-4 probabilities produced by the stochastic map $S^{(R,N_g)}$
acting on an initial $P\po{0}$ with
$\Phi\po{0}=a\po{0}\partial_NP\po{0}$,
$P^{(R,N_g)}=S^{(R,N_g)}P\po{0}$. Then the physical probabilities in the final state
$      \tilde{p}_{\bm  i}^{(R,N_g)}=\sum_{\bm\sigma\in\{0,1\}^\Nb}
P^{(R,N_g)}_{2\bm i+\bm \sigma}$, satisfy
\begin{equation}
  \label{eq:tilpfin}
 \tilde{p}_i^{(R,N_g)}= \lim_{\Nl\to\infty}\hat{\tilde{p}}_i^{(R,N_g,\Nl)}  =|{\Phi}_i^{(N_g)}|/||\, {\Phi}^{(N_g)}\, ||_1\,, 
\end{equation}
where the estimator $\hat{\tilde{p}}_i^{(R,N_g,\Nl)}$ is given by \eqref{eq:ptildeest} with $R_{2\bm i+\bm \sigma}=R_{2\bm i +\bm \sigma}^{(R,N_g,\Nl)}$, the relative frequencies of the byte-4 values for $\Nl$ realizations. 
\end{theorem}
In other words, for $\Nl\to\infty$, the final physical probability distribution on byte-logical values will be given by the Born-1 rule based on the true quantum-mechanical (realified) state $\Phi$, up to the different normalization. 
\begin{proof}
Let
\begin{equation}
  \label{eq:tilp}
\hat{\tilde{p}}_i^{(R,N_g,\Nl)}\equiv \sum_{\bm\sigma}R_{2\bm i +\bm\sigma}^{(R,N_g,\Nl)}
\end{equation}
be the estimate of the physical probabilities after $N_g$ gates based on the relative frequencies for the b4vs with $\Nl$ realizations and including the refresh operations.  By definition of the refresh operations, we have
\begin{equation}
  \label{eq:tilp2}
  \hat{\tilde{p}}_i^{(R,N_g,\Nl')}=|\hat{\psi}_i^{(R,N_g,\Nl)}|\,,
\end{equation}
where in general $\Nl'\ge \Nl$.  % In particular for the second type of refreshments, we can have $\Nl'> \Nl$. 
Now take the limit $\lim_{\Nl\to\infty}$ on both sides of eq.\eqref{eq:tilp2}. On the r.h.s.~we obtain
\begin{align}
  \lim_{\Nl\to\infty}|\hat{\psi}_{\bm i}^{(R,N_g,\Nl)}|&=|\lim_{\Nl\to\infty}\hat{\psi}_{\bm i}^{(R,N_g,\Nl)}|\label{p1}\\
&=|\lim_{\Nl\to\infty}\sum_{\bm \sigma}(-1)^{\sum_j\sigma_j}R_{2\bm i+ \bm \sigma}^{(R,N_g,\Nl)}|\label{p2}\\
 &=|\sum_{\bm \sigma}(-1)^{\sum_j\sigma_j}P_{2\bm i+ \bm \sigma}^{(R,N_g)}|\label{p3}\\
  &=|\psi_{\bm i}^{(R,N_g)}|\label{p4}\\
                                                       &=c^{(N_g)}|\psi_{\bm i}\po{N_g}|\label{p5}\\
 &=d^{(N_g)}|\phi_{\bm i}\po{N_g}|\,.\label{p6}
\end{align}
The crucial step is the one from \eqref{p2} to \eqref{p3}, which
follows from the consistency of the maximum likelihood estimator $R_I$
of the multinomial distribution $P_I$, i.e.~$\lim_{\Nl\to\infty}
R_I=P_I$ \cite{Hoeffding_asymptotically_1965}, and the finite sum 
over $\bm\sigma$ that commutes with $\lim_{\Nl\to\infty}$.  Hence, the
limit in the r.h.s.~of \eqref{p1} exists and is equal to the l.h.s~of
\eqref{p1}. Eq.\eqref{p3} to \eqref{p4} uses the definition \eqref{eq:psiN}. The step from \eqref{p4} to \eqref{p5} is based on the fact
that for $\Nl\to\infty$, $\hat\psi_i\to\psi_i$ $\forall i$, such that
the refreshment based on reconstructing the state from $\hat\psi$
leaves the amplitude ratios of state $\psi$ unaltered. The last step
is the statement of Theorem \ref{thm1}. 
On the l.h.s.~in \eqref{eq:tilp2} we obtain
\begin{align}
    \lim_{\Nl\to\infty}|\hat{\tilde{p}}_i^{(R,N_g,\Nl')}|& =
                                                         \lim_{\Nl'\to\infty}|\hat{\tilde{p}}_i^{(R,N_g,\Nl')}|\label{q1}\\
                                                         & =|\lim_{\Nl'\to\infty} \hat{\tilde{p}}_i^{(R,N_g,\Nl')}|\label{q2}\\
  &= |\lim_{\Nl'\to\infty}\sum_{\bm \sigma}R_{2\bm i+ \bm
    \sigma}^{(R,N_g,\Nl')}|\label{q3}\\
  &= |\lim_{\Nl'\to\infty}\sum_{\bm \sigma}P_{2\bm i+ \bm
    \sigma}^{(R,N_g,\Nl')}|\label{q4}\\
 & =\tilde{p}_{\bm i}^{(R,N_g)}\,.\label{q5}
\end{align}
Comparison of
eqs.\eqref{q5},\eqref{q1}, and \eqref{p6} leads to
eqn.\eqref{eq:tilpfin} when taking into account that from the normalization $||\hat{\tilde{p}}^{(R,N_g,\Nl')}||_1=1$ the constant $d^{(N_g)}$ is
found as $1/||\, \Phi^{(N_g)}\,||_1$.  This completes the proof. 
\end{proof}

While Theorem \ref{thm2} shows that the grabit approach combined with
refreshements leads to a
stochastic emulation of any finite-depth and finite-size quantum
algorithm in the sense that at 
the end of the stochastic process the physical probabilities for any
outcome are given by
the absolute values of the amplitudes of the true quantum mechanical
state (i.e.~the Born-1 rule, implying that we find with high probability
an outcome that also has a high probability for the corresponding
quantum algorithm), the limit
$\Nl\to\infty$ is of course impractical.  It is clear that for completely delocalized
$\psi$, $\Nl$ must be at least of the order $2^\Nb$ to, on average,
even have pairs that annihilate in the refreshment algorithm, or, equivalently, to estimate
$\psi$.    
The refreshments hence cannot, in the limit of large $\Nb$, mend the
problem of the exponentially 
decreasing prefactor of $\psi$ if the contributions to at least one
blv are spread over an exponential number of b4vs, without
exponentially increasing $\Nl$.  
Unfortunately, it appears that 
the majority of useful quantum algorithms use an exponentially large
amount of interference \cite{Braun06,Arnaud07,BraunG08}: they start
with creating a superposition of all computational basis states and
only close the such opened interferometer at the very end, e.g.~with
an inverse QFT, where amplitudes on all other blvs will
contribute with different phases to any $\psi_i$. \\ 

We remark that $N_g$ does not explicitly enter in the proof of theorem \ref{thm2}.  In fact, 
one could do the refreshments at the very end of the algorithm, based
simply on Theorem \ref{thm1}, dealing with the prefactor that
decreases exponentially with $N_g$ by exponentially increasing
$\Nl$. This has the advantage of allowing massive parallelization of
the emulation until the very end, with one thread per realization,
minimal memory of just 2$\Nb$ bits, and basic boolean operations on
those (plus a random number generator 
that generates a random bit for each gate in $\mathcal I$).
At the example of $H^n$ on a single qubit, we have seen,
however, that in practice it is worthwhile to refresh after each gate
$\in \mathcal I$, even for substantial oversampling, in order to
prevent even faster accumulation of errors in $\psi$. For
parallelization this requires re-synchronization of all threads for
each refreshment gate.  
Rather than calculating analytical bounds on $\Nl$ required for a
certain precision of $\tilde p$, we study in the next section
the required scaling of $\Nl$ numerically for the
paradigmatic Quantum Fourier transform and the quantum approximate optimization algorithm (QAOA). 

\section{Quantum algorithms}\label{Sec.qalgo}
We now show how the above emulation of a many-body quantum state and
elementary quantum gates leads naturally to the stochastic emulation
of quantum algorithms on a classical computer.  

\subsection{Quantum Fourier transform}
The quantum Fourier transform on $\Nb$ qubits realizes a unitary transformation with matrix elements 
$U_{ij}$ given in the computational basis as
\begin{equation}
  \label{eq:QFT}
  U_{jk}^{\text{QFT}}=\frac{1}{2^{\Nb/2}}e^{2\pi j\,k/2^\Nb}\,.
\end{equation}
It is a fundamental subroutine for some of the most important quantum
algorithms, such as order finding, Shor's factoring algorithm, and the
hidden-subgroup problem.  It can be efficiently
realized with a
series of Hadamard gates followed by a sequence of phase-gates $R_K$
with
\begin{equation}
  \label{eq:Rk}
  R_k=
  \begin{pmatrix}
    1&0\\
    0& e^{2\pi i/2^k}
  \end{pmatrix}
\end{equation}
on each qubit, see \cite{Nielsen00} Sec.5.1, followed by a sequence of
swap gates that exchange qubit labels 1 and $\Nb$, 2 and $\Nb-1$,
$\ldots$, $\Nb/2$ and $\Nb/2+1$ in the case $\Nb$ is even (or
$(\Nb-1)/2$ and $(\Nb-1)/2+1$ in the case $\Nb$ is odd).  \\

The most basic
use of the QFT is to perform a discrete Fourier transform of
amplitudes of computational basis states, i.e.~a state of the form
$\ket{\psi}=\sum_{j=0}^{2^\Nb-1}x_j \ket{j}$ is mapped to
$\ket{\tilde{\psi}}=\sum_{j,k=0}^{2^\Nb-1}x_jU_{jk} \ket{k}$.  As
such, the QFT is not very useful yet, as the state collapses to a
random computational state when read out, rather than the register
containing a desired
Fourier coefficient $y_k=x_jU_{jk}$, but it becomes useful if the state is
periodic. In that case, the Fourier coefficients will be highly peaked
on some values $k$, and the state will collapse to one of those with high
probability when measured, which allows one to recover the periodicity of the state. 

Since $R_k$ is a complex phase gate, we need to spend a ReIm-grabit to
emulate the QFT, in which case the $R_k$ gate becomes a controlled-rotation
gate of the form \eqref{eq:Uphi} with $\phi=2\pi/2^k$.  Figure
\ref{fig:QFT} shows that the QFT can indeed be emulated stochastically
using grabits.  Panel a.) shows the output state for a periodic input
state of the form $\ket{\psi}=(\ket{0}+\ket{p}+\ket{2p}+\ldots )$ with
$p=3$ for $\Nb=5$ (including the ReIm grabit), $\Nl=10^5$, compared to the exact
state obtained from propagating the input state with the unitary
\eqref{eq:QFT}. One expects a peak at ``wavenumber'' $2\times 16/3$,
(where the 2 is from mapping to a real wavefunction by using the ReIm
grabit) and integer multiples, but since this is not an integer, the
peak gets distributed over 
neighboring wavenumbers.  Panel b.)  is for $p=4$ for $\Nb=6,
\Nl=10^5$.  Here one clearly sees the expected peaks at
wavenumber $2\times 32/4$ and integer multiples.  
\begin{figure}[t]
  \includegraphics[width=0.49\textwidth]{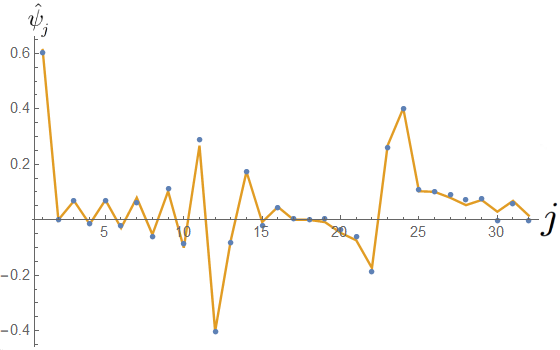}
  \includegraphics[width=0.49\textwidth]{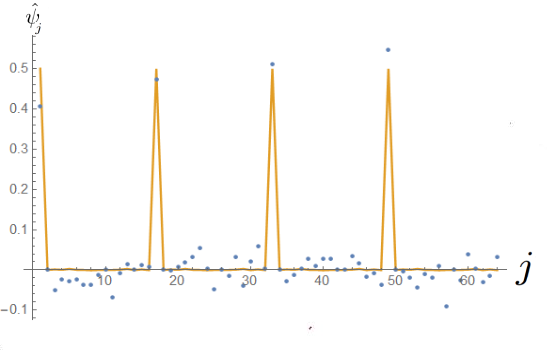}

  \caption{Comparison of the final state estimate $\hat{\psi}$ 
    of QFT from a stochatic
    emulation of the algorithm (blue dots) with exact propagation $\Phi$
    (orange lines; values only defined on integer support;
    computational states counted here from 1 instead from 0). $\hat{\psi}$ is renormalized to $||\hat{\psi}||_2 =1$ for the comparison. {\em
      Left:} $\Nb=5$,   
    $\Nl=10^5$, initial periodic state with period 3 (over the first 4
    grabits).
   {\em Right:} $\Nb=6$ including the ReIm-grabit,
    $\Nl=10^5$, initial periodic state with period 4 (over the first
    5 grabits).
    \label{fig:QFT} }
  % From QFT_20190404.nb, QFT tests
\end{figure}
An important question is how the $\Nl$ needed to find the correct
period of a periodic function via the QFT scales with $\nb$. This is
analysed in Fig.~\ref{fig:QFTscaling}, where the inverse QFT is used
to calculate the period of a state $\ket{\tilde{k}} = \frac{1}{\sqrt{2^{\nb}}} \bigotimes_{j=1}^{\nb} (\ket{0} + \exp(2\pi i k/2^j) \ket{1})$ prepared in the
Fourier basis, 
so that $\text{QFT}^{-1} \ket{\tilde{k}} = \ket{k}$, where 
$k>2^{\nb-1}$ was chosen \footnote{For simplicity, the SWAP operations were
  skipped both in the forward and backward QFT.}. 
If $\text{argmax}_l\,\hat{\tilde{p}}_l=k$, 
after running the algorithm that run is considered  a success.
The figure shows that for a fixed success ratio (i.e.~ratio of number of successful runs to the total number of runs), $\Nl$ scales
exponentially with $\nb$. A fit to $\Nl=a \exp (b\, \nb)$ with $a,b$ as fit parameters generates the scaling law $\Nl = 3.46  \cdot \exp(0.7\nb)
\simeq 3.46\cdot 2^\nb$. 
\begin{figure}[t]
  \includegraphics[width=0.6\textwidth]{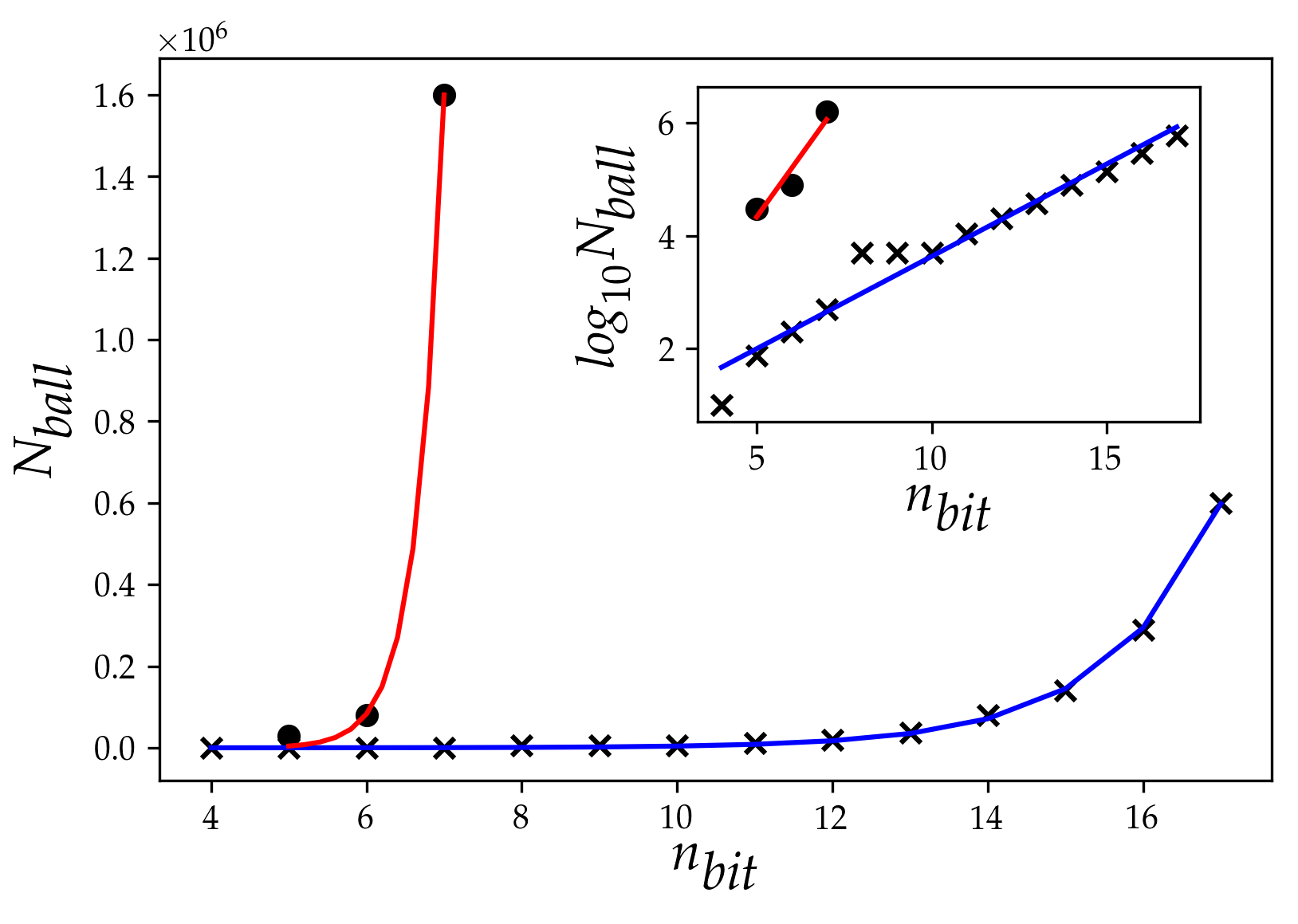}
  \caption{Minimum $\Nl$ needed in order to successfully run an
    inverse QFT in at least $10\%$ of trials as function of the size
    of the circuit in terms of $\nb$. 
    The number of gates in
    each circuit scales like $\mathcal{O}(\nb^2)$. The blue line with crosses shows
    a  fit of the form $\Nl=a \exp (b\, \nb)$ with refreshment $Rf_3$
    after each gate while the red one with points
    shows a similar fit without refreshment. The fits result in an
    exponential scaling of $\Nl \propto \exp( 0.7 \nb)$ for the case
    with refreshment and $\Nl \propto \exp(2.9 \nb)$ without
    refreshment, 
    i.e.~refreshments improve the scaling of the minimum $\Nl$ significantly, but
    still require exponentially large $\Nl$. The inset shows the same
    data in a     lin-log plot.
    \label{fig:QFTscaling} }
\end{figure}

The culprit for this scaling is the estimation of the quantum state
needed in the refreshment. So while refreshments help to fight
the exponential decrease of the amplitude of the state with the number
of gates and quickly become necessary when the number of
interference-creating gates increases with the depth of the algorithm,
they still require in general an exponentially large number of
realizations. For algorithms that are
$p$-blocked  one might envisage refreshments only for
the correlated blocks by keeping track of the blocks of correlated
qubits as proposed in \cite{Jozsa03}.  In that case the grabit
emulation should be efficient, but not more so than the method
in \cite{Jozsa03}.

\subsection{Quantum Approximate Optimization Algorithm}

The Quantum Approximate Optimization Algorithm \cite{farhi_quantum_2014} (QAOA) offers approximate solutions
 for combinatorial optimization problems. The cost function of the optimization problem can
 be translated to the 
 expectation value of a cost Hamiltonian $H_C$. This Hamiltonian combined with a mixer Hamiltonian $H_B$ creates the variational ansatz of the QAOA,
  \begin{equation*}
     \ket{\psi_p(\bm{\gamma},\bm{\beta})} = e^{-i\beta_pH_B} e^{-i\gamma_p H_C}... e^{-i\beta_1 H_B} e^{-i\gamma_1 H_C} \ket{+}^{\otimes \nb},
 \end{equation*}
 where $p$ is an input parameter that sets the 
 depth of the ansatz. A classical optimization over
 $(\bm{\gamma},\bm{\beta})$ aims at minimizing the expectation value
 $\bra{\psi_p(\bm{\gamma},\bm{\beta})}H_C\ket{\psi_p(\bm{\gamma},\bm{\beta})}$,
 the resulting ``cost function'' of the optimization problem.   
 For $p=1$ and some regular graphs of bounded degree, the cost
 function can be calculated analytically
 \cite{farhi_quantum_2014,ryan-anderson_quantum_2018}. In all other
 cases, one needs to estimate
 $\bra{\psi_p(\bm{\gamma},\bm{\beta})}H_C\ket{\psi_p(\bm{\gamma},\bm{\beta})}$
 by sufficiently large sampling,  i.e.~a large number of ``shots'' of
 the quantum algorithm is required for many settings in order to
 perform the classical optimization step. Additional samplings are
 required in the end to estimate the
 ground state energy or to find the ground state with a given fixed
 success probability.  It is then interesting to compare the total
 number of shots needed  with $\Nl$ in the grabit approach, 
 as $\Nl$ has a similar meaning as the number of samples drawn. 
 This comparison is shown in Fig.~\ref{fig:QAOA} for QAOA designed to
 solve a portfolio optimization problem (see Appendix
 \ref{app:B}),  
 where we plot 
the probability $P_{gs}$ to find the system in the true ground state
of $H_C$ after the optimization as function of the number of shots
used, and the same quantity for the grabit approach as function of
$\Nl$.  The number of shots used was calculated by simulating the
quantum algorithm in form of a quantum circuit using the density
matrix simulation software QASM \cite{qiskit}, see Fig. \ref{app:circ}
for an example circuit with $p=1$. Each data point is the mean of 100
optimization runs. The spread in the QASM case is due to different
chosen values for the number of shots per step and the number of steps
in total. We further showcase the dependency 
of the estimation error of the final state compared to the exact quantum state
when simulating an exemplary QAOA circuit for different values of
$\Nl$ as function of $n_\text{gate}$.  The latter figure shows that
the estimate of the grabit state indeed approaches the true state for
high enough $\Nl$. However, we also see in the left figure that for
$\Nl$ equal to the number of shots in the quantum algorithm, the
latter still outperforms the grabit approach slightly. For larger $p$
the circuit depth increases and even larger values $\Nl$ are needed
for high accuracy, whereas in the quantum case the number of shots
increases only due to additional optimization steps. 
It also becomes clear that both quantum circuit and grabit approach
are by far outperformed by the simplest classical approach of going
through all possible combinations, of which there are just $2^5=32$ in
the present problem.

\begin{figure}[t]
  \includegraphics[width=0.49\textwidth,trim={0 0cm 0 0},clip, height=5cm]{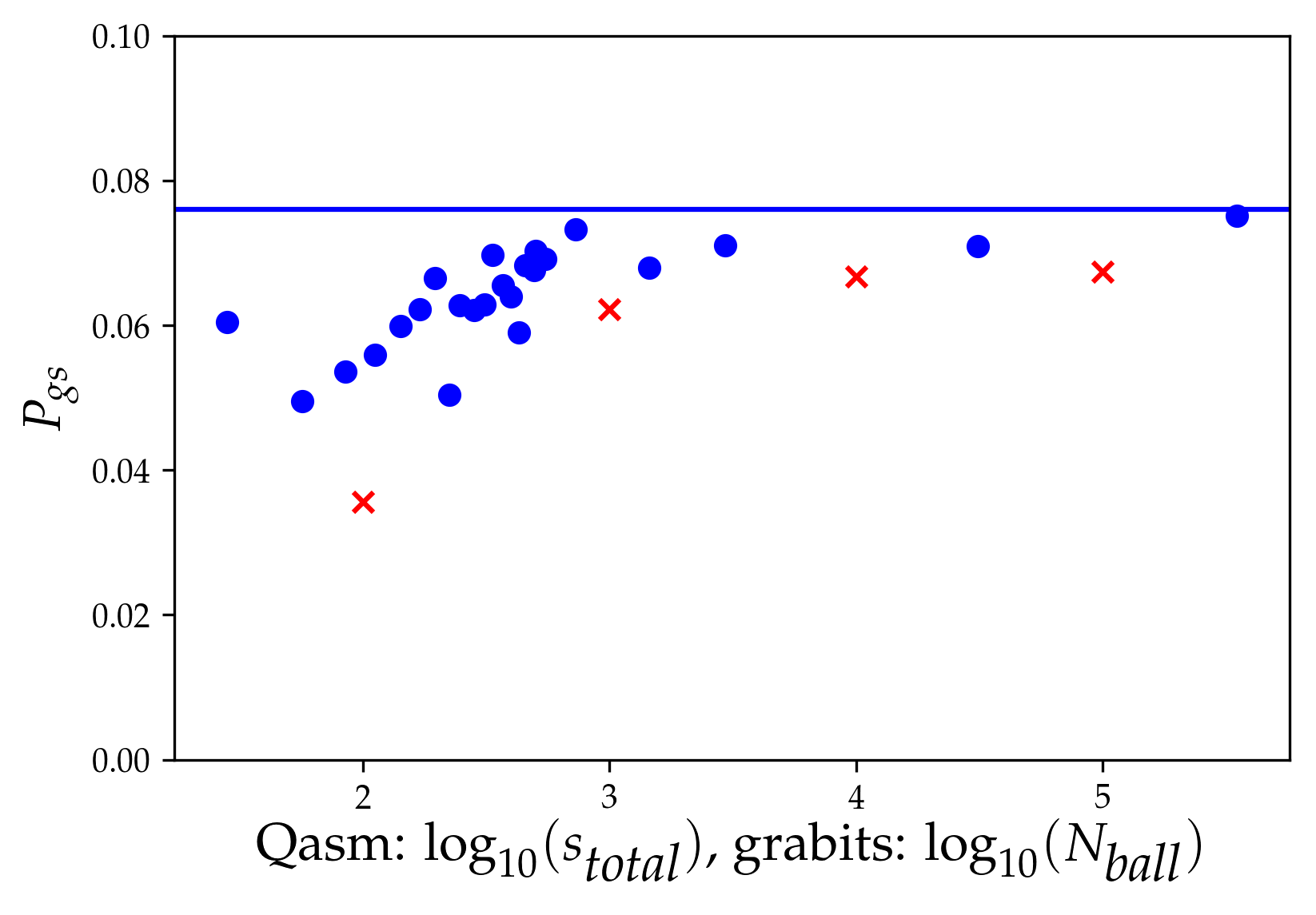}
  \includegraphics[width=0.49\textwidth,trim={0 0 0 0},clip,height=5cm]{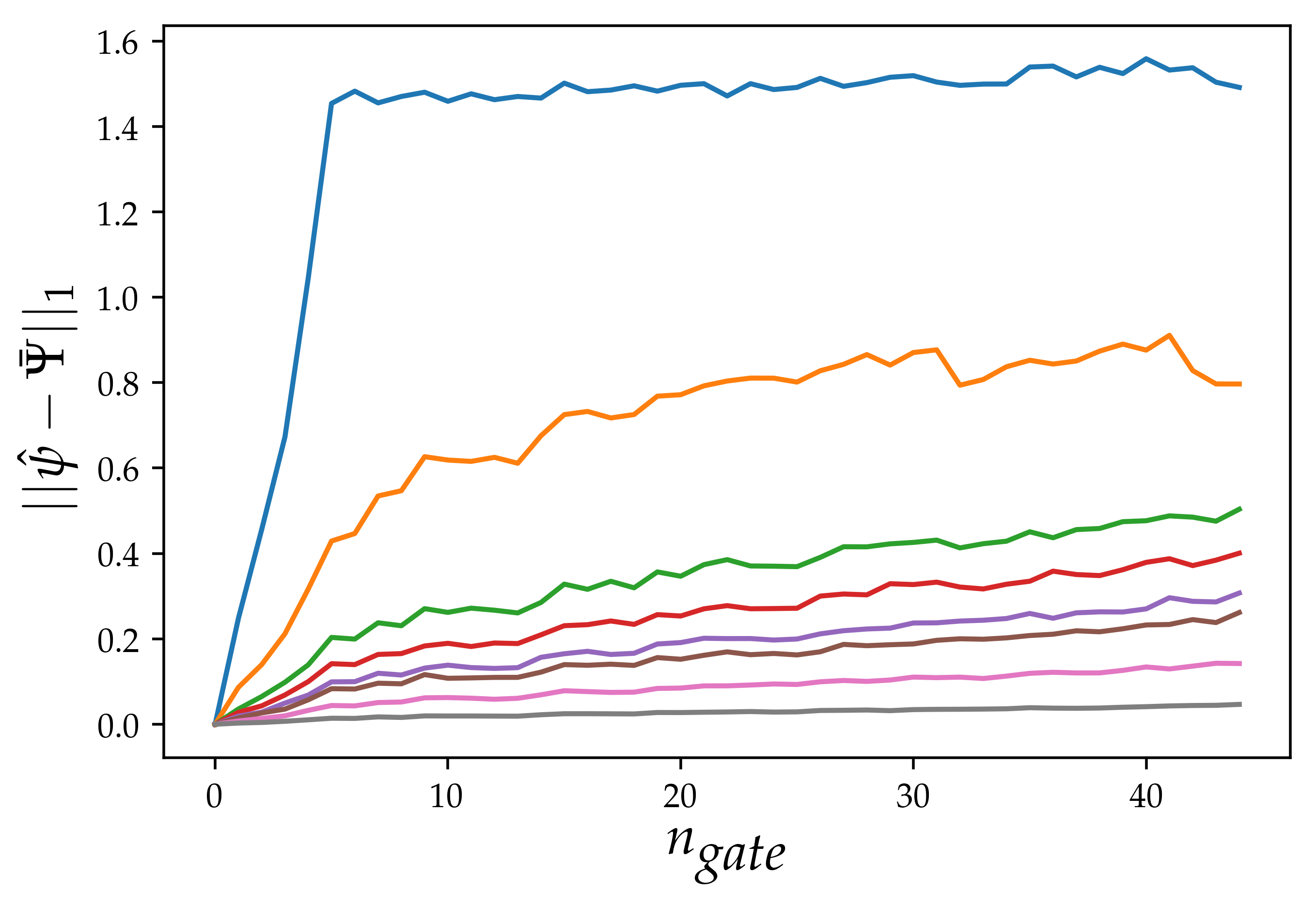}
  \caption{
    Comparison of the performance of the grabit approach and direct
    implementation of QAOA (as simulated with QASM) for a portfolio
    optimization problem  on $\nb=5$ qubits with $p=1$ (see Appendix
    \ref{app:B}). {\em Left:} The ground state probability
    $P_{gs}$ after running the QAOA circuit as a function of the total
    number of shots for the  QASM 
    simulator,  and  of the value of $\Nl$ for the grabit approach.
    The points show the results for the QASM simulator while the crosses represent the grabit approach.  Each shot/sample resembles one evaluation of the cost function. The horizontal blue line shows the result of a statevector simulation. 
    {\em Right:} $||\hat{\psi} - \bar{\Psi }||_1$, with $\bar{\Psi}= (|\Psi_0|^2,...,|\Psi_{2^{\nb}-1}|^2)$,
    as a function of the number
    of executed gates $n_{\text{gate}}$ in a single QAOA circuit for $p=1$, see Fig. \ref{app:circ}, for
    different values of $\Nl=10,100,500,1000,2000,3000,10000,100000$ (top to bottom). $\Nl \gg 2^{\nb}=32$ is required to get a result with a low deviation from the true state.
    \label{fig:QAOA} }
\end{figure}

\section{Conclusions and outlook}\label{Sec.Realism}
We have shown that any multi-qubit state $\Psi$ of $\nb$ qubits can be
represented as a discretized partial derivative $\psi$ of order
$\Nb=\nb+1$ of a classical probability distribution over $\Nb$ ``byte-4 values'', each
represented by 2 classical stochastic bits.  The state $\psi$ is fully
capable of interference in the exponentially large Hilbert space
spanned by all $\psi$ and can be propagated with a stochastic
algorithm that fully reproduces the quantum state evolution of
any pure-state quantum algorithm  up to a prefactor linked to
different normalization.  Complex $\Psi$ can be mapped to real ones by
adding a single additional qubit (resp.~grabit).
We have identified
the stochastic maps for all gates in a universal gate set, and hence
found an automated translation of any quantum algorithm represented by
a quantum circuit to a classical stochastic algorithm that becomes
exact in the limit of $\Nl\to\infty$ (see Theorem\ref{thm1}). 
Based on this, we emulated
several key quantum algorithms.  The norm of the estimator $\hat\psi$
of the grabit state typically decays
exponentially with the depth of the quantum algorithm.   This
can be remedied to some extent with a ``refreshment procedure'' that
exploits a gauge degree of freedom in the representation of $\psi$
as a grabit state, but the number of realizations needed for given
success probability of the studied grabit-algorithms still scales exponentially with
$\Nb$.  This is due to the fact that in the refreshment procedure the
quantum state is estimated, and it can only be estimated based on the
drawn realizations.  This prevents automated  turning 
efficient quantum algorithms into {\em efficient} stochastic
emulations with this method. 

From a foundational perspective, the grabit approach demonstrates a
fundamental difference between how probabilities are propagated in
quantum mechanics and in classical 
mechanics: in the latter, we can only do so by acting on drawn
realizations, whereas in quantum mechanics it appears that we can propagate
states that fully determine the measurement statistics in the
limit of very many measurements when manipulating a single quantum
system --- i.e.~apparently without the material support of all the actual
instances that would be needed classically to precisely implement the
stochastic gates (see
also \cite{mohan_quantum_2021} for an assessment of the ``immaterial
reality'' of a quantum mechanical state). Fundamentally it is unclear
how Nature stores a state then, and where the apparent true randomness
comes from.  A natural speculation would be that the information is
stored in other dimensions, and that the particles that we see in our
world are random samples from many different realizations from that
hidden part of the Universe.  In fact, as mentioned, the present grabit approach can
be generalized further: we can add coordinates $y_i$ to each particle,
and define a wavefunction
$\psi(x_1,\ldots,x_N)=\partial_{y_1}\ldots\partial_{y_N}P(x_1,y_1,\ldots,x_N,y_N)|_{\bm
  y=0}$. After the discretization of the derivative we saw indeed two
different types of coordinates arise naturally, but play a very
different role: byte-logical values $i$, and gradient values
$\sigma$. We can think of the latter as  pointing in the direction of
a ``hidden world'', whereas the former correspond to the usual
coordinates that we see in ``our'' world.  True randomness would then
not be incompatible with realism, a view also advanced in
\cite{gisin_is_2010,wetterich_quantum_2021}, and such a view of reality resembles the
Many Interacting World approach to quantum mechanics (\cite{PhysRevX.4.041013,herrmann_ground_2018} and references therein). However, the number of 
worlds would have to be infinite to exactly represent a quantum
state. \\ 
Theories with extra dimensions  have, of course, a long history in
physics, from Kaluza-Klein till string theories
\cite{PhysRevLett.95.161601}.  Typically the extra dimensions are
treated on an equal footing as the ``normal'' ones, just curled up to very small
length-scales.  This is also true for proofs that they cannot exist,
e.g.~because they would lead to the instability of the hydrogen atom
\cite{lesch_dimensionen_2007}.  In the grabit context the extra
dimensions play a fundamentally different role.  In particular, a
particle cannot disappear into these other dimensions.  Rather we can
think of the extra coordinate of something like an additional
spin-degree of freedom that may have its own dynamics and manifests
itself only if we go to microscopic length-scales, where quantum
inteferences start to play a role.  \\ 

Thus, while the grabit approach has, so far, not lead to new efficient
classical algorithms, it 
might give new insight into the nature of
quantum states, which  almost 100 years after the discovery of quantum mechanics
is still subject of debate and research (see
e.g.~\cite{pusey_reality_2012}).  In particular, our new approach has proven 
that at least one other physical object besides the familiar
many-particle quantum states of quantum mechanics exists that allows
for quantum interference in an exponentially large tensor-product
Hilbert space, combined with the possibility of local
operations, and a probability interpretation attached.  Hopefully this
will lead to the discovery of other such  
objects and fresh insights into possible ``roads to reality'' in  quantum mechanics and
quantum mechanical states in particular.  \\

While finishing the manuscript, a preprint with some ideas similar to
those presented here appeared \cite{wetterich2021quantum}.  %%TODO CITE IS EMPTY

\acknowledgements{DB thanks Michael Hall for
  introducing him to the Interacting Many 
  World (IMW) picture, and Uzy
  Smilansky for encouragements to think about classical entanglement
  more than 20 years ago. He thanks Roderich Tumulka for interesting
  discussions and pointing out additional references to the (IMW)
  picture of Quantum Mechanics. We thank Gerhard Hellstern and Thomas
  Wellens for sharing unpublished work. We acknowledge use of
  the IBM Q for this work. The views expressed here are those of the authors and do
not reflect the official policy or position of 
IBM or the IBM Q team.}\\

{\bf Author contributions}\\
DB conceived the idea, developed the
theory, implemented basic simulations of the quantum algorithms with
the exception of QAOA, and
wrote most of the manuscript.  RM implemented the QFT and QAOA with
QISKIT and as grabit algorithms and performed the scaling analysis of
the required $\Nl$.

\section{Appendix A: Analytical results for the error after $H^2Rf_1$} \label{app:A}
In the following we study the scaling of the error as function of
$\Nl$ 
including refreshments after each $H^2$ sequence on a single qubit. 
Starting with initial state $\ket{0}$ of a single qubit, $H^{2n}$ leaves
the state real in the computational basis, and it is enough to analyse
the propagation of a single grabit. $P_I$ denote here the probabilities for b4vs after the first $H^2$, i.e.~$\bm P=(1/2,0,1/4,1/4)^t$. At that stage, the unnormalized grabit-state estimate is $\hat\psi=(R_0,R_2-R_3)^t$. In the limit of large $N$, a refreshment gate preserves the amplitude ratios of $\hat\psi$ but scales up the prefactor by making the state interference-free.  Hence, after $H^2Rf_1$, the physical probabilities are estimated  as 
\begin{equation}
  \label{eq:pt0}
  \hat{\tilde{p}}_0^{(H^2Rf_1)}
  =\frac{R_0}{R_0+|R_3-R_2|}\,.
\end{equation}
To be specific we focus on $Rf_1$, but any other refreshment routine that fulfills \eqref{eq:1.3} and \eqref{eq:1.4} is covered. We now show that the estimate \eqref{eq:pt0} is efficient, i.e.~$\langle \hat{ \tilde{p}}_0^{(H^2Rf_1)} \rangle \stackrel{\Nl\to\infty}{\longrightarrow} \tilde{p}_0^{(H^2Rf_1)}$, where the expectation value is over the statistical ensemble created by the stochastic propagation, and estimate the error as function of $\Nl$.

Since we expect $|R_3-R_2|\ll R_0$ for $\Nl\to\infty$, the expectation
value  can be expanded as $1-\langle 
\frac{|R_2-R_3|}{R_0}\rangle+{\cal O}(\frac{|R_2-R_3|}{R_0}^2)$, and
the second order will be neglected. The probability for realizing a
certain histogram $\bm R$ is given by the multinomial distribution,
\begin{equation}
  \label{eq:ph}
  p_{\bm R}=N! \prod_{I=0}^3\frac{P_I^{N R_I}}{(N R_I)!}\,,
\end{equation}
where $N\equiv \Nl$ for short, and $R_3=1-R_0-R_1-R_2$. Now use the
Stirling approximation for large $N$, 
$N!\simeq \exp(N \ln N-N)\sqrt{2\pi N}$, consider $R_I$ to vary continuously, and change summation to 
integrals, $\sum_{n_i=0}^N\rightarrow N\int_0^1 dR_i$. Thus, for large $N$, the probability of realizing a given 
histogram $\bm R\equiv (R_0,R_1,R_2,1-R_0-R_1-R_2)$ reads
\begin{equation}
  \label{eq:ph2}
  p_{\bm R}\simeq \frac{1}{2\pi
    N}\frac{1}{\sqrt{R_0R_2(1-R_0-R_2)}}\exp\left[
    N\left\{
\sum_{I=0,2} R_I\ln
\left(\frac{P_I}{R_I}\right) 
+(1-R_0-R_2)\ln\left(\frac{1-P_0-P_2}{1-R_0-R_2}\right) 
    \right\}
  \right]\,.
\end{equation}
The expectation value $\langle Q(\bm R)\rangle\equiv \langle\frac{|R_3-R_2|}{R_0}\rangle$
then becomes
\begin{equation}
  \label{eq:Qh}
 \langle Q(\bm R)\rangle=\frac{N}{2\pi}\int_0^1
  dR_0\left[\int_0^{(1-R_0)/2}dR_2-\int_{(1-R_0)/2}^{1-R_0}dR_2\right]\frac{1-R_0-2R_2}{R_0}\frac{e^{N
  R(\bm R)}}{\sqrt{R_0R_2(1-R_0-R_2)}}\,,
\end{equation}
where the ``action'' $NR(\bm R)$ is given by the argument of $\exp$ in
\eqref{eq:ph2}. We solve the second integral by saddle point
approximation (SPA) in
the limit of large $N$.  This requires finding the maximum of $R(\bm
x)$, i.e.~solving $\partial_{R_2}R(\bm R)=0$, with the
solution $R_2^s=(P_2/(1-P_0))(1-R_0)$. After the gates $H^2$,
$P_2/(1-P_0)=1/2$.  Therefore, the saddle point ${R}_2^s=(1-R_0)/2$
is exactly on the integration boundary for $R_2$ in both integrals. In
addition the pre-exponential factor (considered here as function of
$R_2$ only for the purpose of the integral over $R_2$) $\Phi(R_2)\equiv
(1-R_0-2R_2)/\sqrt{R_0R_2(1-R_0-R_2)}$ vanishes at
$R_2={R}_2^s$. This is an unusual situation for the SPA.  The 1st
problem is still treated in text books such as \cite{Copson65}, see
p.37. We generalize this method to treating also $\Phi({R}_2^s)=0$ by
expanding it. Looking at the 1D case with $x\equiv R_2$ as integration
variable, consider the integral
\begin{equation}
  \label{eq:I}
  I=\int_\alpha^\beta dx \Phi(x)e^{N h(x)}\,.
\end{equation}
Define $t$ by $h(x)=h(\beta)-t^2$. So at $x=\beta$, $t=0$, and more generally,
$x=h^{-1}(h(\beta)-t^2)$, and $dx=-2t/h'(x)\,dt$.  By the mean value
theorem, there exist values $\xi,\xi_1$ such that
$dx=\sqrt{-2h''(\xi)}/(-h''(\xi_1))dt$ \cite{Copson65}.  For $N\to\infty$, these move
arbitrarily close to $\beta$, such that in this limit
$dx=-dt/\sqrt{-h''(\beta)/2}$, where a maximum of $h(x)$ at $x=\beta$
was assumed, i.e.~$h(x)$ increases monotonically sufficiently close to
$x=\beta$ for $x<\beta$. Together with $h'(\beta)=0$, this allows one
to write $t^2=h(\beta)-h(x)\simeq -h''(\beta)(\beta-x)^2/2$ and
$\Phi(x)\simeq \Phi'(\beta)(x-\beta)$. Inserting this into
the integral, we find
\begin{equation}
  \label{eq:I2}
  I\simeq \int_0^\tau dt\frac{\Phi'(\beta)(x-\beta)}{\sqrt{\frac{|h''(\beta)|}{2}}}e^{N(h(\beta)-t^2)}\,,
\end{equation}
where the value of $\tau>0$ is irrelevant, and
$x-\beta=\sqrt{-2/h''(\beta)}t$.  With this we 
finally obtain from the lower boundary at $t=0$ 
\begin{equation}
  \label{eq:I3}
  I\simeq \frac{\Phi'(\beta)e^{Nh(\beta)}}{Nh''(\beta)}\,.
\end{equation}
Applied to \eqref{eq:Qh}, and considering that both integrals over
$R_2$ give the same value when the minus sign of the second is
included, we obtain
\begin{equation}
  \label{eq:Qh2}
 \langle Q(\bm R)\rangle=\frac{1}{\pi}\int_0^1dR_0\frac{e^{N\tilde{R}(R_0)}}{R_0^{3/2}}\,,\,\,\,\tilde{R}(R_0)=R(R_0,{R}_2^s(R_0))\,,
\end{equation}
which leads to
\begin{equation}
  \label{eq:tilp0}
  \langle\hat{\tilde{p}}_0\rangle\simeq 1-\frac{2}{\sqrt{\pi N}}+{\cal O}(\langle Q^2(\bm
  R)\rangle,\frac{1}{N})\,. 
\end{equation}
The expectation of the estimator of the logical probability vector after the sequence $H^2Rf_1$ hence 
deviates from  the correct value $(1,0)$ as
\begin{equation}
  \label{eq:diffp}
  ||
  \begin{pmatrix}
    \langle\hat{\tilde{p}}_0\rangle\\
    \langle\hat{\tilde{p}}_1\rangle\\
  \end{pmatrix}
  -
\begin{pmatrix}
    1\\
    0\\
  \end{pmatrix}
  ||_2\simeq 2 \sqrt{\frac{2}{\pi}}\frac{1}{\sqrt{N}}\simeq \frac{1.596}{\sqrt{N}}\,,
\end{equation}
in reasonable agreement with the numerically estimated
behavior $1.4/\sqrt{N}$, with the error averaged over 10 realizations of the algorithm and $\Nl$ ranging from 50 to 5000, when considering the relatively large statistical uncertainty from the 10 runs. 
When iterating the sequence as $(H^2Rf_1)^n$, the random deviation of 
$\hat{\tilde{\bm p}}$ from the exact value (1,0) is expected to lead to a diffusive behavior of the 2-norm 
\eqref{eq:diffp} as function of $n$, i.e.~a slow increase
$\simeq\sqrt{1.596 n/N}$. Since in   
Fig.\ref{fig:expinc} $n_H$ includes  refreshment gates after each $H$,
which are, however, irrelevant after the 1st Hadamard and all
following ones with odd number, thus diluting the increase of the
2-norm by a
factor $\sqrt{3/4}$, one expects a behavior $\simeq\sqrt{1.596 \times
  (3/4) n_H/N}\simeq 0.01\sqrt{n_H}$.  Given the relatively large
statistical uncertainty, this is again in reasonable
agreement with the numerically found $\exp(-5.08413 + 0.532838 \ln
)\simeq 0.0064 n_H^{0.538}$.

\section{Appendix B: Measure and bounds of destructive interference}\label{app:Destr}
\subsection{Cram\'er-Rao bound for the estimation of the grabit state}
After concentration of signs, $\psi_{\bm i}$ and $\tilde{p}_{\bm i}$ can be written as
\begin{eqnarray}
  \label{eq:psipt}
  \psi_{\bm i}& =& P_{2\bm i}-P_{2\bm i+1}\\
  \tilde{p}_{\bm i}& =& P_{2\bm i}+P_{2\bm i+1}\,,
\end{eqnarray}
where $2\bm i+1$ means $2\bm i+(0,\ldots,0,1)$. Hence, the $\psi_{2\bm i+1}$ parametrize the b4v probabilities as
\begin{equation}
  \label{eq:pp}
  P_{2\bm i}=\frac{1}{2}(\tilde{p}_{2\bm i+1}+  \psi_{\bm i}),\,\,\, P_{2\bm i+1}=\frac{1}{2}(\tilde{p}_{2\bm i+1}-  \psi_{\bm i})\,.
\end{equation}
The Fisher information with respect to the $\psi_{\bm i}$ is given by
\begin{eqnarray}
  \label{eq:Ipsi}
  I_{\psi_{\bm i}}&=&\sum_{\bm j,\sigma}P_{2\bm j+\sigma}\left(\frac{\partial \ln P_{2\bm j+\sigma}}{\partial \psi_{\bm i}}\right)^2\\
&=&  \frac{1}{4}\sum_{\bm j,\sigma}'P_{2\bm j+\sigma}^{-1}\,,
 \end{eqnarray}
where the sum $\sum_{\bm j,\sigma}'$ is over all ${\bm j,\sigma}$ with $P_{2\bm j+\sigma}\ne 0$.  Hence the Cram\'er-Rao bound for the smallest possible standard deviation $\text{sdv}(\hat \psi_{\bm i})$ of an arbitrary unbiased estimator $\hat\psi$ of $\psi$ based on $\Nl$ samples reads
\begin{equation}
  \label{eq:sdv}
  \sum_{\bm i}''\text{sdv}(\hat \psi_{\bm i})\ge \sum_{\bm i}'' \frac{2}{\sqrt{\Nl(\frac{1}{P_{2\bm i}}+\frac{1}{P_{2\bm i+1}})}}\,,
\end{equation}
where the $\sum_{\bm j}''$ is over all ${\bm j}$ with $P_{2\bm j}P_{2\bm j+1}\ne 0$. This demonstrates the importance of large $\Nl$ for a good estimate of the grabit state.  Without any refresh, the effective $\Nl$ that contributes to an estimate of $\Nl$ decays due to destructive interference.  Indeed, that decay can be considered a {\em measure of destructive interference}. Let $\bm H$ be a histogram normalized to $\sum_{\bm I}H_{\bm I}=\Nl$. After the action of any gate or even an entire algorithm, we have a new histogram $\bm H'$. The estimate $\hat \psi'$ from that histogram is identical to the one obtained from $\bm H'$ by sign concentration and pair annihilation, called $\bm H''$, even if these two operations are not performed. It leads to a new effective $\Nl'=\sum_{\bm I}H_{\bm I}''$. Its mean value is given by
\begin{eqnarray}
  \label{eq:nl'}
  \langle \Nl'\rangle &=& \sum_{\bm H'} P_{\bm H'}\sum_{\bm i}|\sum_{\bm \sigma,\text{even}}H'_{2\bm i+\bm \sigma}-\sum_{\bm \sigma,\text{odd}}H'_{2\bm i+\bm \sigma}|\\
                      &=&\Nl \sum_{\bm H'} P_{\bm H'}\sum_{\bm i}|\sum_{\bm \sigma,\text{even}}R'_{2\bm i+\bm \sigma}-\sum_{\bm \sigma,\text{odd}}R'_{2\bm i+\bm \sigma}|\,,
\end{eqnarray}
where $P_{\bm H'}$ is the probability for a histogram ${\bm H'}$ given by a multinomial distribution of the b4v probability distribution $P_I'$ after the gate or algorithm.  In the limit of $\Nl\to\infty$, $R_I'\to P_I'$, and the distribution of histograms $R_I$ becomes arbitrarily sharp, concentrated on $P_I'$.  This leads to $\langle \Nl'\rangle \longrightarrow \Nl \sum_{\bm i}  |\psi_{\bm i}^{'(u)}|=\Nl ||\psi^{'(u)}||_1$, where the unnormalized grabit state $\psi_{\bm i}^{'(u)}$ is defined as 
\begin{equation}
  \label{eq:nl'2}
\psi_{\bm i}^{'(u)}\equiv \sum_{\bm \sigma,\text{even}}P'_{2\bm i+\bm \sigma}-\sum_{\bm \sigma,\text{odd}}P'_{2\bm i+\bm \sigma}\,.
\end{equation}
The reduction of $ \Nl\longrightarrow\langle \Nl'\rangle$ depends on the initial state.  E.g.~a first Hadamard on a state $\ket{0}$ does not introduce interference yet, and $\Nl' = \Nl$ in this case.  Only when the interferometer closes, i.e.~after $H^2$, probability amplitudes are brought to overlap and interfere away the amplitude of $\ket{1}$.  By considering the worst case over all input states $\psi$, the following definition of a measure of destructive interference is motivated:
\begin{equation}
  \label{eq:Dmeasure}
  {\cal D}\equiv \max_{\psi}(1-||\psi^{'(u)}||_1/||\psi||_1)\,,
\end{equation}
which  is only a function of the grabit gate (or algorithm) applied. We calculate it in the next subsubsection for the Hadamard gate and the phase gate. 

\subsection{Destructive interference and the decay of the effective $\Nl$}
%\subsubsection{Hadamard gate on a single qubit out of $\Nb$ many}
After applying a Hadamard gate on a single grabit (taken as the first one in the following, w.l.g.), the 1-norm $||\psi'||_1$ of the grabit state reads
\begin{eqnarray}
  \label{eq:n1psi'}
  ||\psi'||_1&=&||H\, \psi||_1=\sum_{i_1,\ldots,i_n}|\psi_{i_1,\ldots,i_n}'|\\
           &=&\frac{1}{2}\sum_{\bar{\bm i}}\left( |\psi_{0\bar{\bm i}}+\psi_{1\bar{\bm i}}|+|\psi_{0\bar{\bm i}}-\psi_{1\bar{\bm i}}| \right)\\
  &=&\sum_{\bar{\bm i}}\max\{|\psi_{0\bar{\bm i}}|,|\psi_{1\bar{\bm i}}|\}\,,
\end{eqnarray}
where $\bar{\bm i}\equiv(i_2,\ldots,i_n)$, $i_1\bar{\bm
  i}\equiv(i_1,i_2,\ldots,i_n)$, and we used $|a+b|+|a-b|=\max\{|a|,|b|\}$ for all $a,b\in \mathbb R$. Then
\begin{eqnarray}
  \label{eq:minPH}
  \min_{\psi}||\psi^{'(u)}||_1&=&\min_{\psi} \sum_{\bar{\bm i}}\max\{|\psi_{0\bar{\bm i}}|,|\psi_{1\bar{\bm i}}|\}\\
&\ge&\min_\psi \max\{ \sum_{\bar{\bm i}}|\psi_{0\bar{\bm i}}|,\sum_{\bar{\bm i}}|\psi_{1\bar{\bm i}}|\}\\
&\ge&\min_\psi \max\{ \sum_{\bar{\bm i}}|\psi_{0\bar{\bm i}}|,||\psi||_1-\sum_{\bar{\bm i}}|\psi_{0\bar{\bm i}}|\}\\
&\ge&\min_{0\le x\le ||\psi||_1}\max\{x,||\psi||_1-x\}\\
  &=&\frac{1}{2}||\psi||_1\,.
\end{eqnarray}
Hence, ${\cal D}_H=1/2$ for an arbitrary number of grabits (with $H$ applied to one of them). 
Similarly, one shows for the phase gate that ${\cal D}_{R_\varphi}=1-\frac{1+q^2}{(1+q)^2}$ with $q=|\cot \varphi|$. 
Note that above we assumed a fixed initial $\Nl$.  After the first gate that introduces destructive interference, the effective $\Nl'$ will be distributed.  However, since for $\Nl\to\infty$ the distribution of histograms becomes arbitrarily sharp, almost all histograms will have $\Nl'=  \langle \Nl'\rangle $. In that limit the effective $\Nl$ will hence decay over an entire quantum algorithm such that the final $\Nl'$ satisfies
\begin{equation}
  \label{eq:NLf}
  \langle\Nl'\rangle \gtrsim \Nl \prod_{g=1}^{N_g}{\cal D}_g\,, 
\end{equation}
where the product is over all gates of the algorithm. Hence, without any refresh the effective remaining $\Nl'$ decays exponentially with the numer of gates that lead to destructive inteference.  It should be kept in mind, however, that \eqref{eq:NLf} is the worst case scenario, pessimized over all intermediate states.  From the experience with $H^{2n}$ one expects that this exponential decay (and corresponding deterioriation of an estimate of the final $\psi$) can be avoided with refreshs after each interference-generating gate $g$, $S_g\in \cI$, but only for $\Nl\gg 2^{\Nb}$. 

\newpage

\section{Appendix C: Portfolio Optimization using QAOA} \label{app:B}

 Mean-variance portfolio analysis as introduced by Markowitz \cite{markowitz} can be used to choose assets out of a stock market in a way that minimizes risk. The following description is based on \cite{notebook_hellster,wellens}. Consider a portfolio consisting of $n$ assets. The return $\mu^i$ of each asset $i$ can be calculated via \begin{equation*}
    \mu^i = (\Pi_{k=1}^{M-1} (1+r_k^i))^{252/(M-1)},
\end{equation*} 

where $r_k^i = (p_{k}^i - p_{k+1}^i)/p_k^i$, $p_{k}^i$ is the historical daily price of asset $i$, $M$ is the number of available price data  points, and 252 is 
the number of trading days in a year. Correlations between different assets can be analysed using the annualized covariance matrix $\sigma$ defined as 
\begin{equation*}
    \sigma_{ij} = \frac{252}{M-1} \sum_{k=1}^{M-1} (r_k^i- \bar{r}^i)(r_k^j - \bar{r}^j),
\end{equation*}

where $\bar{r}^i = 1/(M-1) \sum_{k=1}^{M-1} r_k^i$ is the mean daily price change. The problem of finding a set of  $B$ assets out of a portfolio containing $n$ stocks that minimizes risk while maximizing return can be thought of as minimizing the following cost function

\begin{equation*}
    f(x_1,...,x_n) = q \sum_{i,j=1}^n x_j x_i \sigma_{ij} - (1-q) \sum_{i=1}^n x_i \mu^i 
\end{equation*}

under the constraint $\sum_i x_i = B$, where $x_i =1,0$, and $q$ is 
a factor to scale the weight of return and risk. The lower and upper index $i$ is used to label the assets in consistency with previous equations and is not representing any summation convention.  By replacing $x_i \rightarrow (1-Z_i)/2$ and absorbing constant terms into $\sigma'$ and $\mu'$, we can transform $f$ into the quantum operator 
\begin{equation*}
    H_c = \sum_{i>j} Z_iZ_j \sigma_{ij}' + \sum_i Z_i \mu_i'.
\end{equation*}

Choosing the mixer operator to be $H_B = -\sum_i X_i$  results in the QAOA Ansatz state
  \begin{equation*}
     \ket{\psi_p(\bm{\gamma},\bm{\beta})} = e^{-i\beta_pH_B} e^{-i\gamma_p H_C}... e^{-i\beta_1 H_B} e^{-i\gamma_1 H_C} \ket{+}^{\otimes n}. 
 \end{equation*}

The explicit data used in Fig.~\ref{fig:QAOA} are stock prices of 5 randomly selected stocks in the DAX-30 averaged over a period of 5 years.

\begin{figure}[h]
    \centering
    \includegraphics[width=\linewidth, height=5cm]{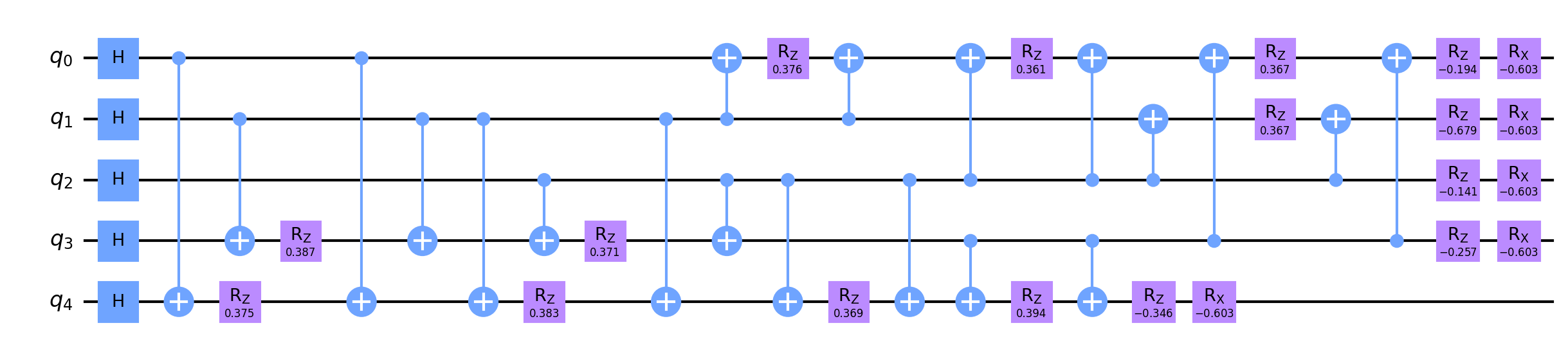}
    \caption{QAOA circuit for $p=1$ used to analyze the development of error $||\hat{\psi} - \Phi||_1$ as a function of the number of gates in Fig. \ref{fig:QAOA}. The $R_Z$ gates and CNOTs implement the cost Hamiltonian $H_c$ while the $R_x$ gates realize the mixer $H_B = -\sum_i X_i$. }
    \label{app:circ}
\end{figure}

%\subsection{Shor's factoring algorithm}

% \bibliography{../bibfile_master/mybibs_bt,grabit}
%\bibliography{C:/Users/dany_/Documents/mypapers/bibfile_master/mybibs_bt}
\bibliography{mybibs_bt}

\end{document}